\documentclass[11pt,reqno]{amsart}

\usepackage{hyperref}
\usepackage{amsmath}
\usepackage{amsthm}
\usepackage{amsfonts}
\usepackage{amssymb}

\usepackage{latexsym,amsmath}
\usepackage{geometry}
\usepackage{fixmath}

\usepackage[dvips]{graphicx}
\usepackage{graphicx}
\usepackage{fullpage}
\usepackage{appendix}

\usepackage{color} 
\usepackage{times}

\usepackage{float}
\usepackage{graphicx}
\usepackage[symbol]{footmisc}
\usepackage[shortlabels]{enumitem}

\newtheorem{theorem}{Theorem}[section]
\newtheorem{lemma}[theorem]{Lemma}

\newtheorem{proposition}[theorem]{Proposition}
\newtheorem{claim}[theorem]{Claim}
\newtheorem{corollary}[theorem]{Corollary}

\newtheorem{definition}[theorem]{Definition}
\newtheorem{observation}[theorem]{Observation}

\newcommand{\cA}{{\mathcal{A}}}

\newcommand{\cE}{{\mathcal{E}}}
\newcommand{\cF}{{\mathcal{F}}}

\newcommand{\cI}{{\mathcal{I}}}

\newcommand{\G}{\mathbold{G}}

\newcommand{\bJ}{{\mathbold{J}}}

\newcommand{\bM}{{\mathbold{M}}}

\newcommand{\bS}{\mathbold{S}}
\newcommand{\bT}{\mathbold{T}}

\newcommand{\J}{{\mathbold{J}}}
\newcommand{\dc}{d_{\rm cond}}

\newcommand{\bInf}{{\mathbold \Inf}}

\newcommand{\bmu}{\mathbold{\mu}}
\newcommand{\bxi}{\mathbold{\xi}}
\newcommand{\bsigma}{{\mathbold{\sigma}}}
\newcommand{\btau}{{\mathbold{\tau}}}

\newcommand{\lnorm}{\left | \left |}
\newcommand{\rnorm}{\right| \right|}

\newcommand{\Inf}{\Gamma}
\newcommand{\mdownup}{\mu_{\downarrow\uparrow}}

\newcommand{\Exp}{\mathbb{E}}
\newcommand{\mExp}[1]{\left\langle#1\right\rangle}

\newcommand{\range}{ \mathrm{im}}

\date{\today}

\begin{document}

\title{Broadcasting with   Random  Matrices}
\author{Charilaos Efthymiou$^*$ and Kostas Zampetakis$^*$}

\thanks{ 
$^*$Research supported by EPSRC New Investigator Award, grant EP/V050842/1,  and 
Centre of Discrete Mathematics and Applications (DIMAP), University of Warwick, UK
}

\address{Charilaos Efthymiou, {\tt charilaos.efthymiou@warwick.ac.uk}, The University of Warwick, Coventry, CV4 7AL, UK.}
\address{Kostas Zampetakis, {\tt kostas.zampetakis@warwick.ac.uk}, The University of Warwick, Coventry, CV4 7AL, UK.}

\maketitle

\begin{abstract}
Motivated by the theory of {\em spin-glasses} in  physics, we study the so-called  {\em reconstruction problem} on the  tree, and on the sparse random 
graph $\G(n,d/n)$. Both  cases reduce naturally to analysing  broadcasting models, where 
each edge  has its own  broadcasting matrix, and this matrix  is drawn independently from a
predefined distribution.

We establish the {\em reconstruction threshold} for the cases where the broadcasting matrices give rise to 
symmetric, 2-spin Gibbs distributions. This  threshold  seems to be a natural extension of  
the well-known {\em Kesten-Stigum bound} that manifests in  the classic version of the reconstruction problem. 
Our results determine, as a special case,  the reconstruction threshold for the prominent {\em Edwards–Anderson model}  
of spin-glasses, on the tree. 

Also,  we extend our analysis to the setting of the Galton-Watson random tree, and the (sparse) random graph $\G(n,d/n)$, where
 we establish the corresponding  thresholds. 
Interestingly, for the Edwards–Anderson model on the random graph, we show that the {\em replica symmetry 
breaking}  phase transition, established by Guerra and and Toninelli in 
\cite{guerra2004high}, coincides with the reconstruction threshold. 

Compared to classical Gibbs distributions, spin-glasses have several unique features. In that respect, 
their study calls  for new ideas, e.g. we introduce novel estimators for the reconstruction problem. 
The main technical challenge in the analysis of such systems,
is the presence of (too) many levels of randomness, which we manage to circumvent by
utilising  recently proposed  tools coming from the analysis of Markov chains.  
\end{abstract}

\section{Introduction}

\noindent
Motivated by the theory of {\em spin-glasses} in  physics, we study the so-called 
{\em reconstruction problem}  with respect to the related distributions, on the tree, and on the sparse random graph $\G(n,d/n)$.

Spin-glasses are {\em disordered} magnetic materials that are studied by physicists (not necessarily the theoretical ones).  
It has been noted that even though they are a type of magnet, actually, ``they are not very good at being magnets''. 
Metallic spin-glasses are ``unremarkable conductors", and the insulating spin-glasses are ``fairly useless as 
practical insulators \ldots ", e.g. see   \cite{SteinNewmanSpinGlassBook}. 

However, the research on spin-glasses has provided   tools to analyse some exciting, 
and extremely challenging,  problems in mathematics, physics, but also {\em real world} ones. 
Through their study, we have garnered a deep understanding of the nature of complex systems. 
 A case in point is the pioneering work of Giorgio Parisi in `70s on the so-called \emph{Sherrington-Kirkpatrick} spin-glass, which 
 introduces the formulation of  the 
renowned {\em replica symmetry breaking} \cite{RSBParisi}. Parisi's ideas  were highly influential in physics community, and
later, in mathematics, and computer science.  The theory of replica symmetry breaking was among the groundbreaking 
ideas  which  got Parisi the Nobel Prize in Physics in 2021. 

Perhaps one of the most successful, and extensively studied spin-glass models, is the famous \emph{Edwards-Anderson} model (EA-model for short), introduced back in `70s by Sam Edwards and Philip Anderson in \cite{EAIsingIntroWork}. 
Few months after the work of Edwards and Anderson,  David  Sherrington
and Scott Kirkpatrick,  in \cite{SKModel}, introduced their own model of spin-glasses,   the well-known in computer science
literature,  {\em Sherrington-Kirkpatrick} model (SK-model for short). As it turns out, the SK-model corresponds to the {\em mean field}  
version of the EA-model.

Given a fixed graph $G=(V,E)$,  the Edwards-Anderson model  with {\em inverse} temperature $\beta>0$, 
is the {\em random} Gibbs  distribution 
$\bmu$ on the configuration space $\{\pm 1\}^V$ defined as follows: let  $\{\bJ_e: {e\in E}\}$ be independent identically distributed  
(i.i.d.) {\em standard Gaussians}. Then each configuration $\sigma\in\{\pm 1\}^V$ receives probability mass $\bmu(\sigma)$, defined by 
\begin{align}\label{eq:IntroGibbs}
\bmu(\sigma) &\propto 
\exp\left( \beta \cdot\sum_{\{u,w\}\in E} {\bf 1}\{\sigma(u)=\sigma(w)\} \cdot \bJ_{\{u,w\}}   \right) \enspace ,
\end{align}
where $\propto$ stands for ``proportional to". We usually refer to   $\{\bJ_e\}_{e\in E}$ as  the {\em coupling parameters}.  Let us comment here that, alternatively, the Gibbs distribution is defined by replacing the indicator ${{\bf 1}\{\sigma(u)=\sigma(w)\}}$ in~\eqref{eq:IntroGibbs}, with the product $\sigma(u)\sigma(w)$. However, the two formulations are equivalent, as a simple transformation converts one to the other (see Appendix \ref{sec:IndProdEquiv}).
We also note that there is a simpler version of the Edwards-Anderson model, in which coupling  parameters take independently~$\pm 1$ values, uniformly at random.

Apart from its mathematical elegance, and theoretical importance, the  Edwards-Anderson model, and the  related spin-glass distributions,
arise also in  applications such as   neural networks (e.g.  the so-called Hopfield model),  protein folding, and conformational dynamics. 
We refer the interested reader to \cite{SteinNewmanSpinGlassBook}, and references therein.

In this work, we  largely study the Edwards-Anderson model on  trees, and the (locally tree-like)   {\em random graph} $\G(n,d/n)$
with constant expected degree $d$.  This is the random graph on $n$ vertices, such that each edge appears independently with probability  $d/n$.  
Since the Edwards-Anderson model on $\G(n,d/n)$ shares essential features with random  {\em Constraint  Satisfaction Problems} ($r$-CSPs for short), it is not surprising that has been studied extensively in terms of phase transitions, in physics, e.g. \cite{franz2001exact,mezard2006reconstruction}, mathematics, e.g.  \cite{guerra2004high,CoEfJKKCMI}, but also in computer science, 
 e.g. for  sampling algorithms \cite{EfthICALP22,AlaouiMS22}.

In contrast  to the  standard  Gibbs distributions on trees, e.g. the Ising model, the Hard-core model, and the Potts model, 
the Edwards-Anderson model, despite  being the most  basic distribution for spin-glasses, has not been sufficiently studied.  
As a result, several fundamental  questions about it still remain open.
Here, we consider the tree {\em reconstruction problem} for the Edwards-Anderson model (and some natural extensions).

The {reconstruction problem} studies  the effect of the configuration at a vertex $r$, on that of the 
vertices at distance $h$ from $r$, as  $h\to\infty$.  Specifically, we want to distinguish the region of parameters where 
the effect is vanishing,  from that where the effect is non-vanishing. Typically, the two regions are specified in terms of a 
{\em sharp threshold}, i.e., we have an  abrupt transition from one region to the other as we vary the parameters of the model.  We usually call this 
phenomenon   {\em reconstruction threshold}, and it has been the subject of intense study, e.g. 
\cite{MolloyJACM, AchCojeBarriers,HiguchiReconIsing,BlehRuiZagrebNonReconIsing,SlyReconColouring,CoEftRSA15HCGeom}. 
In the context of r-CSPs,  the onset of reconstruction  has been linked to an abrupt  deterioration of the performance of algorithms
(both searching and counting), e.g. see \cite{AchCojeBarriers}. 

In this work, among other results, we establish precisely the reconstruction threshold for the Edwards-Anderson model 
 on the $\Delta$-ary tree, the Galton-Watson tree with general offspring distribution, and the random graph $\G(n,d/n)$. 
 Furthermore, as far as the Edwards-Anderson model on  $\G(n,d/n)$ is concerned, 
we combine our results with \cite{guerra2004high,CoEfJKKCMI}, to conclude that the reconstruction threshold coincides with 
the so-called {\em Replica Symmetry Breaking} phase transition.

Interestingly, for the $\Delta$-ary tree, we establish the reconstruction threshold, not only for the Edwards-Anderson model, but also
for the general version of the Gibbs distribution $\bmu$ defined in \eqref{eq:IntroGibbs}. That is, the coupling parameters
 are i.i.d. following a {\em general distribution}, not necessary the standard Normal.  

It turns out that  the corresponding reconstruction problems on the Galton-Watson tree with $\mathrm{Poisson}(d)$ offspring, 
and on the sparse random graph $\G(n,d/n)$, are not too different from each other.  Connections have been established 
 between these two Gibbs distributions,  e.g. see  
 \cite{BapCoOgEftSIlentPlanting, CokrzPerZdeSTOC17,CoEfJaaLocalWeakConv,CoKaMullCPC}. 
We relate the two reconstruction results, i.e., for the tree and the graph,  by exploiting the 
idea of planted-model (Teacher-Student model \cite{zdeborova2016statistical})  and the notion of  {\em contiguity} \cite{CoEfJKKCMI}. 
 In that respect, our basic analysis involves the complete  $\Delta$-tree, and the Galton-Watson tree, while,  subsequently, 
 we extend these results to  the random graph $\G(n,d/n)$. 

We study the reconstruction problem on trees by means of the broadcasting models. 
These are abstractions of  {\em noisy transmitted} information over the edges of the 
tree, i.e., the edges act as noisy channels.  To our knowledge, the study of the broadcasting models,
 and the closely 
related reconstruction problem, dates back to `60s with the seminal 
work of Kesten and Stigum  \cite{KSBound}.

Establishing the reconstruction threshold for the Edwards-Anderson model on the $\Delta$-ary tree, as well
as the generalisation of this distribution, turns out to be a challenging problem.  
The difficulty of these models stems from the manifestation of 
local {\em frustration phenomena}, i.e., mixed ferromagnetic and antiferromagnetic interaction in the same neighbourhood,
but also from the ``many levels of randomness" we need to deal with in their analysis.

To this end, we make an extensive use of various potentials in order to simplify the analysis. 
To establish non-reconstruction, we employ 
some newly introduced techniques in the area of Markov chains and Spectral Independence \cite{OptMCMCIS,VigodaSpectralInd}, 
that combine  potential functions to analyse tree recursions. To establish reconstruction, we use a carefully crafted potential 
as an estimator for the root configuration. We call this estimator {\em flip-majority vote}.


\subsection{Broadcasting, Reconstruction and the Kesten-Stigum  bound}
Consider the $\Delta$-ary tree $T=(V,E)$, of height $h>0$. Let $r$ be the root
of the tree $T$.
Broadcasting  on  $T$,  is a stochastic process which abstracts noisy transmission of information over the edges of
the tree.

There is a finite set of spins $\cA$, and an $\cA\times \cA$ stochastic 
matrix $M$, which we call the {\em broadcasting matrix}, or {\em transition matrix}. 
With the broadcasting we obtain a configuration $\bsigma\in \cA^V$ by working recursively  as follows: assume that the 
configuration at the root $r$ is obtained according to some predefined  distribution  over $\cA$.  If for the non-leaf
vertex $u$ in $T$ we have  $\bsigma(u)=i$, then for each vertex $w$, child of $u$,  we have  $\bsigma(w)=j$ with
probability $M(i,j)$, independently of the other children, i.e.,
\begin{equation*}
\Pr[\bsigma(w)=j\ |\ \bsigma(u)=i] = M(i,j) \enspace .
\end{equation*}
Here we  assume that $\bsigma(r)$ is distributed uniformly at random in $\cA$.

A natural problem to study in this setting is the so-called {\em reconstruction problem}. 
Suppose that $\mu_{h}$ is the marginal distribution of the configuration  of  the vertices at distance $h$ from the root. 
The reconstruction problem amounts to studying the influence of the configuration at the root of the tree 
to the marginal $\mu_h$.
Specifically, we want to compare the two distributions $\mu_{h}(\cdot \ |\ \bsigma(r)=i )$, and $\mu_{h}(\cdot \ |\ \bsigma(r)=j)$
for different $i,j\in \cA$,  i.e., $\mu_h$ conditional on the  configuration at the root being $i$ and $j$, respectively.
The comparison is by means of  the  total variation distance, i.e., 
\begin{equation*}
\lnorm \mu_{h}(\cdot \ |\ \bsigma(r)=i ) -\mu_{h}(\cdot \ |\ \bsigma(r)=j )  \rnorm_{\rm TV} \enspace.
\end{equation*}
Typically, we focus on the behaviour of the quantity above,  as $h$ grows. 
\begin{definition}\label{def:DetermReconstruction}
We say that the distribution $\mu$ exhibits {\em reconstruction} if there exist spins $i,j\in \cA$  such that
\begin{equation*}
\limsup_{h\to\infty} \lnorm \mu_{h}(\cdot \ |\ \bsigma(r)=i ) -\mu_{h}(\cdot \ |\ \bsigma(r)=j )  \rnorm_{\rm TV}>0 \enspace.
\end{equation*}
On the other hand, if for all $i,j\in \cA$ the above limit is zero, then we have  {\em non-reconstruction}. 
\end{definition}

The broadcasting process we describe above gives rise  to  well-known Gibbs distributions  on $T$ such as the 
{\em Ising model}, the {\em Potts model} etc.  In terms of the Gibbs distributions on the tree, the reconstruction problem can be formulated as to 
 whether the  free-measure on the tree is {\em extremal}, or not.  The extremality  here is considered with respect to
whether  the Gibbs distribution  can be  expressed as  a convex  combination of two, or more measures, e.g.
see \cite{Georg88}.   
It is interesting to compare the extremality condition with various spatial mixing conditions of the Gibbs distribution. 
Perhaps the most interesting case is to compare it with  the Gibbs tree {\em uniqueness}. Then, it is standard to
show that the extremality is a {\em weaker}  condition than uniqueness.

The reconstruction problem has been studied since 1960s. Perhaps the most general result in the area 
is the so-called  {\em Kesten-Stigum bound} \cite{KSBound}, or KS-bound for short.  
Let 
$\Delta_{\rm KS}=\Delta_{\rm KS}(M)$ be such that
\begin{equation}\label{def:DetermDeltaKS}
\Delta_{\rm KS}=\lambda_2^{-2}(M) \enspace, 
\end{equation} 
where $\lambda_2(M)$ is the second largest, in magnitude, eigenvalue of the transition matrix $M$.  
The result of \cite{KSBound} implies that  if  $\Delta > \Delta_{\rm KS}$, then we have reconstruction. 

In light of the above, a natural question is whether  the condition   $\Delta <\Delta_{\rm KS}$  implies  that we have non-reconstruction.  
In general, the answer to this question is no, 
e.g. see  \cite{BhatSlyTetHCRecon,SlyReconColouring}. 
However, for several important distributions, including the Ising model, the KS-bound is tight,  in the sense
that the condition  $\Delta <\Delta_{\rm KS}$ indeed implies non-reconstruction,  see 
\cite{BlehRuiZagrebNonReconIsing,EvKenPerSchulConConIsingTree,HiguchiReconIsing}.


\newcommand{\StochMat}{\bS}

\subsection{Broadcasting with random matrices}
Here, we consider the natural problem of  broadcasting on a tree, where  the transition matrix is {\em random}. 
In this setting, as before, we consider the $\Delta$-ary tree $T=(V,E)$, of height $h>0$, rooted at $r$. Also, we have 
a finite set of spins $\cA$.  
Rather than using the same matrix for every edge of the tree, each edge has its own matrix, which is an independent 
sample from a predefined distribution $\psi$. 

More formally, every  $\cA\times \cA$ stochastic matrix can be viewed as a point in the ${|\cA|^2}$ Euclidean space.
We endow  the set of all $\cA\times \cA$ stochastic matrices with the  $\sigma$-algebra induced by the Borel algebra.
Then,  $\psi$ is a distribution over the set of these matrices.

Once we have a matrix for each edge of $T$, the broadcasting proceeds with  the same rules as in the deterministic 
case. If  for the non-leaf vertex $u$ in $T$ we have  $\bsigma(u)=i$, then the vertex $w$, child of $u$,  gets $\bsigma(w)=j$ with
probability $\bM_e(i,j)$, independently of the other children of $u$, i.e.,
\begin{equation*}\nonumber
\Pr[\bsigma(w)=j\ |\ \bsigma(u)=i] = \bM_{e}(i,j) \enspace,
\end{equation*}
where $e=\{u,w\}$.

The above setting gives rise to a {\em random} probability measure on the set of configurations $\cA^V$
which we   denote as $\bmu=\bmu_{T, \psi}$. Hence,  the configuration $\bsigma\in \cA^V$ we get from 
the broadcasting,  consists  of {\em two-levels of randomness}. The first level is due to the fact that the measure $\bmu$ 
is  induced by the random instances of  the broadcasting matrices $\{\bM_e\}_{e\in E}$.  Once these matrices 
have been fixed, the second level of randomness emerges from the random choices of the broadcasting process. 
The above formulation gives rise to well-studied Gibbs distributions, such as the Edwards–Anderson model of spin-glasses,
by choosing appropriately the distribution  $\psi$.

In this new setting,  we study the reconstruction problem. Here, the definition of reconstruction differs slightly from Definition \ref{def:DetermReconstruction} above.  Denote with $\bmu_{h}$ the marginal of $\bmu$ on the vertices at distance $h$ from the root of the tree $T$. Then, the reconstruction problem is defined as follows:

\begin{definition}\label{def:RandomReconstruction}
For a distribution  $\psi$  on $\cA\times \cA$ stochastic matrices , we say that the random measure
$\bmu=\bmu_{T,\psi}$  exhibits {\em reconstruction} if there exist spins $i,j\in \cA$  such that
\begin{equation*}\nonumber
\limsup_{h\to\infty} \mathbb{E} \left [ \lnorm \bmu_{h}(\cdot \ |\ \bsigma(r)=i ) -\bmu_{h}(\cdot \ |\ \bsigma(r)=j )  \rnorm_{\rm TV} \right]>0 \enspace,
\end{equation*}
where the expectation is with respect to the randomness of $\bmu$.

On the other hand, if for all $i,j\in \cA$ the above limit is zero, then we have  {\em non-reconstruction}. 
\end{definition}
We consider the reconstruction problem in terms of the KS-bound, i.e., we examine whether it is tight, or not. 
Before addressing this question, we need to specify what the parameter $\Delta_{\rm KS}$ might be in this setting.

It turns out that  a natural candidate for $\Delta_{\rm KS}$ can be defined as follows:
Let $\bM$ be a matrix sampled from the distribution $\psi$, and define 
\begin{equation}\label{def:MatrixXi}
\Xi=\mathbb{E}\left [\bM\otimes\bM\right] \enspace,
\end{equation} i.e., 
the matrix $\Xi$ is the expectation of the tensor product of the matrix $\bM$ with itself.  Let 
${\bf 1} \in \mathbb{R}^{\cA}$ denote the vector whose entries are all equal to one.  
Also, write 
\begin{equation*}
\cE = \left\{ z\in \mathbb{R}^{\cA}\otimes \mathbb{R}^{\cA}: \forall y\in \mathbb{R}^{\cA}\  
 \langle z, {\bf 1}\otimes y\rangle=\langle z, y\otimes{\bf 1}\rangle=0\right\}\enspace,
\end{equation*}
 where  $\langle \cdot, \cdot \rangle$ is the standard
inner product operation.    
Then, we define $\Delta_{\rm KS}(\psi)$ to be such that
\begin{equation}\label{def:GlassDeltaKS}
\Delta_{\rm KS}(\psi) =
\left( \max_{x\in \cE: \lnorm x \rnorm=1}\langle \Xi x, x\rangle\right)^{-1} \enspace.
\end{equation}
The above quantity, $\Delta_{\rm KS}$, arises in the study of phases transitions in  random CSPs   \cite{CoEfJKKCMI}. 
Specifically,  it signifies an  {\em upper bound}  on the  density of the so-called    {\em Replica Symmetric}  phase, of symmetric
Gibbs distributions.  The  value  $\Delta_{\rm KS}$ is derived in \cite{CoEfJKKCMI} by means  of a stability analysis of  
the so-called  {\em free-energy}  functional. 
Note that the above definition for $\Delta_{\rm KS}(\psi)$  applies to  any set of spins $\cA$, and any distribution $\psi$ on 
$\cA\times \cA$ matrices.

Here, we prove that the above is indeed the analogue of KS-bound for {\em symmetric},  2-spin distributions $\bmu$ (including the EA model). That is, for any distribution $\psi$ over the broadcasting matrices whose support is comprised 
of symmetric $2\times 2$ matrices, we  prove that the $\Delta$-ary tree $T$ exhibits reconstruction 
when $\Delta>\Delta_{\rm KS}(\psi)$, while we have non-reconstruction when $\Delta<\Delta_{\rm KS}(\psi)$. 

Furthermore, we go beyond the  basic case  of  the $\Delta$-ary tree.  Firstly,   we extend our results to the cases
where the underlying graph is the {\em Galton-Watson} random tree with general offspring distribution. Secondly, 
we exploit the notion of contiguity 
of measures to derive non-reconstruction results for the Edwards-Anderson model on the random graph $\G(n,d/n)$.

\section{Results}

We start the presentation of our results on the 2-spin, symmetric distributions, by considering the $\Delta$-ary tree. 
Specifically, for integers $\Delta>0$ and $h>0$, let $T=(V,E)$ be the $\Delta$-ary tree of height $h$, rooted
at vertex $r$.   We  let $\cA=\{\pm 1\}$ be the set of spins.

Suppose that we have a broadcasting process on $T$, while
assume that each edge of the tree is equipped with its own broadcasting  matrix, each matrix drawn {\em independently} from the distribution induced by the following experiment:
We have  two parameters, a real number $\beta> 0$, and a distribution $\phi$ on the real numbers $\mathbb{R}$, i.e., we have the probability space $(\mathbb{R}, \cF, \phi)$ where
$\cF$ is the $\sigma$-algebra induced by the Borel algebra. We generate a matrix $\bM$ following the two steps below:
\begin{description}
\item[Step 1]  Draw $\bJ\in \mathbb{R}$ from the distribution $\phi$.
\item[Step 2]  Generate the $\cA\times \cA$ matrix $\bM$ such that
\begin{align}\label{def:BdMatixM}
\bM&= \frac{1}{\exp( \beta \bJ)+1}\left [
\begin{array}{cc}
\exp( \beta \bJ)  & 1  \\ \vspace{-.3cm} \\
1 &\exp( \beta \bJ) 
\end{array}
\right ] \enspace. 
\end{align}
\end{description}
Note that our broadcasting matrices are always {\em symmetric}.

The above broadcasting process   gives rise to  configurations  in $\cA^V$ following the Gibbs distribution $\bmu_{\beta, \phi}$  specified as follows:
Let $\{\bJ_e\}_{e\in E}$ be independent, identically distributed (i.i.d.) random variables such that each 
one of them is  distributed as in $\phi$ (this is the same distribution used to generate matrix $\bM$).  
Each $\sigma \in \cA^V$   is assigned probability mass $\bmu_{\beta,\phi}(\sigma)$ defined by 
\begin{equation}\label{def:SpinGlass}
\bmu_{\beta,\phi}(\sigma) \propto \textstyle \exp\left( \beta\sum_{\{w,u\}\in E} {\bf 1}\{\sigma(u)=\sigma(w)\} \cdot \bJ_{\{u,w\}}   \right) \enspace ,
\end{equation}
where $\propto$ stands for ``proportional to".

At this point, it is immediate that by choosing $\phi$ to be the {\em standard  Gaussian} distribution, we retrieve the
 Edwards-Anderson model in \eqref{eq:IntroGibbs}. Note however, that \eqref{def:SpinGlass} above generates a whole family of ``spin-glass" 
distributions with the EA-model  being  a special case.

The definition of the distribution of the broadcasting matrix in \eqref{def:BdMatixM} allows us to derive  an explicit formula for  the quantity $\Delta_{\rm KS}$  in  \eqref{def:GlassDeltaKS}. Specifically,  for  $\bJ$ distributed according to $\phi$, it is not hard to prove (see Appendix \ref{sec:KSMatrix}) that
\begin{equation}\label{eq:DeltaKSBPhi}
\Delta_{\rm KS}(\beta,\phi) = 
\textstyle \left( \mathbb{E}\left [   \left (\frac{1-\exp(\beta \bJ)}{1+\exp(\beta \bJ)}  \right)^2 \right] \right)^{-1}\enspace, 
\end{equation}
where the expectation is with respect  to the random variable $\bJ$.
In light of the above, we prove the following result for the general Gibbs distribution.

\begin{theorem}\label{thrm:GeneralDeltary}
 For  a real number $\beta> 0$, and a distribution $\phi$ on the real numbers $\mathbb{R}$
let $\Delta_{\rm KS}=\Delta_{\rm KS}(\beta,\phi)$ be defined as in \eqref{eq:DeltaKSBPhi}.

For any integer $\Delta>\Delta_{\rm KS}$, the Gibbs distribution $\bmu_{\beta, \phi}$, defined as in \eqref{def:SpinGlass},
on the $\Delta$-ary tree
exhibits reconstruction. On the other hand,   if  $\Delta < \Delta_{\rm KS}$ the distribution $\bmu_{\beta,\phi}$ 
exhibits non-reconstruction. 
\end{theorem}

The proof of Theorem \ref{thrm:GeneralDeltary} appears in Section \ref{sec:thrm:GeneralDeltary}.
Let us state the implications of Theorem \ref{thrm:GeneralDeltary} for the 
Edwards-Anderson model on the $\Delta$-ary tree. 

\begin{corollary}\label{cor:EAIsingDeltaAry}
For $\beta>0$ and the standard Gaussian $\bJ$,  let 
\begin{equation*}
\Delta_{\rm EA}(\beta) =\textstyle \left( \mathbb{E}\left [   \left (\frac{1-\exp(\beta \bJ)}{1+\exp(\beta \bJ)}  \right)^2 \right] \right)^{-1}\enspace, 
\end{equation*}
where the expectation is with respect to  $\bJ$.  

For any integer $\Delta>\Delta_{\rm EA}(\beta)$, the distribution $\bmu_{\beta}$, the Edwards-Anderson model with inverse temperature 
$\beta$ on the $\Delta$-ary tree, 
exhibits reconstruction.
On the other hand,   if  $\Delta < \Delta_{\rm EA}(\beta)$ the distribution $\bmu_{\beta}$ 
exhibits non-reconstruction. 
\end{corollary}

\subsection{The case of the Galton-Watson tree}

As a further step, we study the reconstruction problem on the Galton-Watson tree. Even though
this is a very interesting problem on its own, we make use of our results for the Galton-Watson tree to derive subsequent results for $\G(n,d/n)$, see Section \ref{sec:ReconCSPs}.

Let  $\zeta:\mathbb{Z}_{\geq 0}\to [0,1]$ be a distribution over the non-negative integers. 
Then, the rooted tree $\bT$ is a Galton-Watson tree with offspring distribution $\zeta$, if
the number of children for each vertex in  $\bT$ is distributed according to $\zeta$, {\em independently} from the other vertices.

Note that  broadcasting with random matrices over the Galton-Watson tree $\bT$, gives rise to configurations that 
consist of {\em three} levels of randomness. 
One of the challenges we circumvent with our analysis, is to  disentangle all of three levels of randomness,
and make clear the contribution of each one of them.  
Before getting there, we need to clarify what we mean 
by (non-)reconstruction in the current  setting. 

\begin{definition}\label{def:GWRandomRecon}
Consider the  distributions $\phi$ over $\mathbb{R}$ and   $\zeta$ over $\mathbb{Z}_{\geq 0}$, and  a real number $\beta \ge 0$.
Let   the Galton-Watson tree $\bT$ with offspring distribution $\zeta$,  while let  the measure $\bmu=\bmu_{\beta,\phi}$ be defined as in
\eqref{def:SpinGlass},   on the tree $\bT$. 
We say that  $\bmu$  exhibits {\em reconstruction} if 
\begin{equation*}
\limsup_{h\to\infty}  \mathbb{E}_{\bT } \left[  \ 
\mathbb{E}_{\bmu } \left [  \lnorm \bmu_{h}(\cdot \ |\ \bsigma(r)=+1) -\bmu_{h}(\cdot \ |\ \bsigma(r)=-1)  \rnorm_{\rm TV} \ | \ \bT \right] 
\ \right]>0 \enspace. 
\end{equation*}
On the other hand, if  the above limit is zero, then we have  {\em non-reconstruction}. 
\end{definition}
For the above, recall that $\bmu_{h}$ is the marginal of $\bmu$ on the set of vertices at distance $h$ from the root. 
Note that if $\bT$ has no vertex at level $h$, then the total variation distance above is, degenerately, equal to zero. 
We use the double expectation in Definition \ref{def:GWRandomRecon} for the sake of clarity: 
we can just replace it by a single expectation with respect to both the random tree $\bT$, and the random measure $\bmu$. 

As far as the reconstruction problem on the Galton-Watson trees is concerned, we have the following result, which we prove in Section \ref{sec:thrm:GeneralGWtree}.

\begin{theorem}\label{thrm:GeneralGWtree}
 For any real numbers $d>0, \beta>0$, for any distribution $\phi$ on  $\mathbb{R}$,  for any distribution 
 $\zeta$ on $\mathbb{Z}_{\geq 0}$  with expectation $d$ and {\em bounded second moment}, let $\bT$ be the Galton-Watson tree with offspring distribution $\zeta$.  Let also $\bmu_{\beta,\phi}$ be the Gibbs distribution defined as in \eqref{def:SpinGlass},  on the tree $\bT$. Finally, let $\Delta_{\rm KS}=\Delta_{\rm KS}(\beta,\phi)$ be defined as in \eqref{eq:DeltaKSBPhi}.

 The distribution $\bmu_{\beta,\phi}$  exhibits reconstruction if $d>\Delta_{\rm KS}$. On the other hand,   
 if  $d < \Delta_{\rm KS}$,  the distribution $\bmu_{\beta,\phi}$ exhibits non-reconstruction. 
\end{theorem}
 
Let us now state the implications of Theorem \ref{thrm:GeneralGWtree} 
for the  Edwards-Anderson model on the Galton-Watson tree. 

\begin{corollary}\label{cor:EAIsingGW}
For $\beta>0$, consider the quantity $\Delta_{\rm EA}(\beta)$ defined in Corollary \ref{cor:EAIsingDeltaAry}.
 For  any real number $d>0$,  and any distribution $\zeta:\mathbb{Z}_{\geq 0}\to [0,1]$ with expectation 
 $d$, and {\em bounded second moment},   let $\bT$ be the Galton-Watson tree with offspring distribution $\zeta$.
 
 Then, for  $\bmu_{\beta}$ the Edwards-Anderson model with inverse temperature $\beta$, on the tree $\bT$, the following
 is true. The distribution $\bmu_{\beta}$  exhibits reconstruction if $d>\Delta_{\rm EA}(\beta)$. On the other hand,   
 if  $d < \Delta_{\rm EA}(\beta)$,  the distribution $\bmu_{\beta}$ exhibits non-reconstruction. 
\end{corollary}

\subsection{The Edwards-Anderson model on \texorpdfstring{$\G(n,d/n)$}{Lg}} \label{sec:ReconCSPs}

For integer $n\geq1$, and real $p\in[0,1]$, let $\G=\G(n,p)$ be the random graph on $V_n=\{x_1,\ldots,x_n\}$,
whose edge set $E(\G)$ is obtained by including each edge with probability, $p$ {\em independently}.

The {\em Edwards-Anderson model}  on $\G$ at inverse temperature $\beta>0$,   is defined as follows: 
for $\J = \{\J_e\}_{e\in E(\G)}$ a family of independent {\em standard Gaussians},  we let 
	\begin{equation}\label{eq:kSpin1}
	\mu_{\G, {\J},\beta}(\sigma) \textstyle =\frac{1}{Z_\beta(\G,\J)} \exp\left( \beta \sum_{x\sim y}{\bf 1}\{\sigma(y)=\sigma(x)\}\cdot \J_{\{x,y\}}  \right) \enspace,
	\end{equation}
where
	\begin{equation*}
	Z_\beta(\G, {\J})  \textstyle =\sum_{\tau\in\{\pm1\}^{V_n}} \exp\left( \beta \sum_{x\sim y}{\bf 1}\{\tau(y)=\tau(x)\}\cdot \J_{\{x,y\}}  \right)\enspace.
	\end{equation*}
Here we assume that $p=\frac{d}{n}$, where $d>0$ is a fixed number. Typically, we study this distribution as $n\to \infty$. 
The natural question we ask here is how does the model change as we vary $d$.
According to the physics predictions, for any $\beta$ there exists a {\em condensation threshold}, denoted as $\dc(\beta)$, where the
function  
\begin{equation*}
d\mapsto\lim_{n\to\infty}\frac1n\Exp[\ln Z_\beta(\G,{\J})]
\end{equation*} 
is non-analytic~\cite{franz2001exact}. This conjecture was proved  by Guerra and Toninelli~\cite{guerra2004high}.
The regime $d<\dc(\beta)$ is called the {\em replica symmetric phase}.  This region has several interesting properties; here we consider one that seems to be most relevant to our discussion. For any $d<\dc(\beta)$ the distribution $\mu_{\G,\J,\beta}$ satisfies the following property: for $\bsigma$ distributed as in 
$\mu_{\G,\J,\beta}$, for   two randomly chosen vertices ${\bf x}$ and ${\bf y}$, the configurations $\bsigma({\bf x})$ and $\bsigma({\bf y})$ 
are asymptotically independent. 
Formally, the above can be expressed  as follows: for $d<\dc(\beta)$ and any $i,j\in \{\pm 1\}$,  we have that 
\begin{equation*}
\limsup_{n\to\infty}\frac{1}{n^2}\sum_{x,y\in V_n} \Exp\left[\mExp{{\bf 1}\{\bsigma(x)=i\}\times {\bf 1}\{\bsigma(y)=j\}}-
\mExp{{\bf 1}\{\bsigma(x)=i\} }\times \mExp{ {\bf 1}\{\bsigma(y)=j\}}\right] =0\enspace,
\end{equation*}
where $\mExp{\cdot}$ denotes expectation with respect to the Gibbs distribution  $\mu_{\G,\J,\beta}$.
Note that the above holds not only for pairs of vertices, but also for sets of $k$ vertices,  for any fixed integer $k>0$. 
Using  our notation, the work  by Guerra and Toninelli~\cite{guerra2004high} implies the following result. 
\begin{theorem}[\cite{guerra2004high}]\label{thrm:DcondValueEAIsingG}
For any $\beta>0$, for the distribution $\mu_{\G,\J,\beta}$ defined as in \eqref{eq:kSpin1},  we have that
\begin{align*}
\dc(\beta) &=  \textstyle \left( \Exp\left[  \left (\frac{1-\exp(\beta \bJ)}{1+\exp(\beta \bJ)}  \right)^2 \right] \right)^{-1} \enspace,
\end{align*}
where $\bJ$ is a standard Gaussian random variable. 
\end{theorem} 
Interestingly, one obtains the above  by combining  our Theorem \ref{thrm:GeneralGWtree} and  using 
standard results from  \cite{CoEfJKKCMI,Coja-OghlanGGRS22}.
Our main focus is on the reconstruction threshold for the Edwards-Anderson model on $\G$. The reconstruction for
$\mu_{\G,\J,\beta}(\cdot)$ is defined in a slightly different way  than what we have for the random tree. 

\begin{definition}\label{def:GnpRecon}
For $d>0$, for $\beta>0$,  consider the Gibbs distribution 	$\mu_{\G,\J,\beta}$ as this is defined in \eqref{eq:kSpin1}.
We say that the  measure $\mu=\mu_{\G,\bJ, \beta}$  exhibits {\em reconstruction} if 
\begin{align}\nonumber
\limsup_{h\to\infty} \lim_{n\to\infty} \frac{1}{n}\sum_{x\in V_n}\mathbb{E}  \left [  \lnorm \mu_{x,h}(\cdot \ |\ \bsigma(x)=+1 ) -\mu_{x,h}(\cdot \ |\ \bsigma(x)=-1 )  \rnorm_{\rm TV} 
\ \right]>0 \enspace,
\end{align}
where $\mu_{x,h}$ denote the Gibbs marginal  at the vertices at distance $h$ from vertex $x$. 
On the other hand, if the above limit is zero, then we have  {\em non-reconstruction}. 
\end{definition}

Perhaps, it is interesting to notice the order with which we take the double limit in the above definition.

Furthermore, we let the reconstruction threshold, denoted 
as $d_{\rm recon}$, to be the infimum over $d>0$ such that 
\begin{align*}
 \limsup_{h\to\infty}  \lim_{n\to\infty}\frac{1}{n}\sum_{x\in V_n}\mathbb{E}  \left [  \lnorm \mu_{h}(\cdot \ |\ \bsigma(x)=+1 ) -\mu_{h}(\cdot \ |\ \bsigma(x)=-1 )  \rnorm_{\rm TV} 
\ \right]>0 \enspace.
\end{align*}
The region of values of $d$ such that $d<d_{\rm recon}$ is called the {\em non-reconstruction phase}. It is immediate from
Definition \ref{def:GnpRecon} that, for any $d<d_{\rm recon}$, we have that non-reconstruction.

In the following result, we prove that  the replica symmetric phase coincides with the non-reconstruction phase of
the Edwards-Anderson model on $\G$. 
\begin{theorem}
For any $\beta>0$,  for the distribution $\mu_{\G,\J,\beta}$ defined as in \eqref{eq:kSpin1}, we have that $d_{\rm recon}(\beta)=\dc(\beta)$.
\end{theorem}
The above follows from Theorems \ref{thrm:DcondValueEAIsingG}, \ref{cor:EAIsingGW} and   \cite[Corollary 1.5]{CoEfJKKCMI}.

\subsubsection*{Notation} For the graph $G=(V,E)$ and the Gibbs distribution $\mu$ on the set of configurations
$\{\pm 1\}^V$.  For a configuration $\sigma$,  we let $\sigma(\Lambda)$ denote the configuration that $\sigma$
specifies on the set of vertices $\Lambda$.
We let   $\mu_{\Lambda}$  denote the marginal of $\mu$ at the  set $\Lambda$.
We let $\mu(\cdot \ |\ \Lambda, \sigma)$, denote the distribution 
$\mu$ conditional on the configuration at $\Lambda$ being $\sigma$. Also, we interpret the conditional
marginal $\mu_{\Lambda}(\cdot \ |\ \Lambda', \sigma)$, for $\Lambda'\subseteq V$, in the natural way.

\section{Approach}\label{sec:Approach}

A major challenge in our setting is that we have to deal with multiple levels of randomness, i.e., we have 
two levels of randomness in the case of the $\Delta$-ary tree,  while the levels increase with the Galton-Watson trees or $\G(n,d/n)$. 
To circumvent this problem,   we follow an analysis that allows us to disentangle the different sources of randomness
in our models.  In this section, we provide a high-level description of our approach. We restrict our discussion on the 
$\Delta$-ary tree.   

\subsection*{Non-reconstruction}
Consider the $\Delta$-ary tree $T=(V,E)$ rooted at $r$.  
Suppose that we have a distribution $\mu$ as in \eqref{def:SpinGlass} on $T$, while assume that each  edge $e\in E$ has its own coupling parameter $J_e$. 
Assume, for the moment,  that the  coupling parameters at the edges are fixed, e.g. the reader may assume 
that are arbitrary real numbers. That is, each $J_e$ can be either positive, or negative. 
Hence, one might consider the aforementioned distribution as a {\em non-homogenous} Ising model
which involves both ferromagnetic and anti-ferromagnetic interactions. 
Let us focus on  non-reconstruction.  We derive an upper bound on 
\begin{equation*}
\lnorm \mu_{h}(\cdot \ |\ \bsigma(r)=+1 ) -\mu_{h}(\cdot \ |\ \bsigma(r)=-1 )  \rnorm_{\rm TV}\enspace, 
\end{equation*}
which is expressed in terms of the {\em influence} between neighbouring vertices. The notion of influence between
vertices is the same as  the one developed in the context of {\em Spectral  Independence} 
technique for establishing rapid mixing of Glauber dynamics \cite{OptMCMCIS, VigodaSpectralInd}. 
These influences are used  in the context of the so-called  {\em down-up} coupling to establish non-reconstruction.  
This is a coupling approach  from \cite{BhatVeraVigodaWeitz}, which also relies on ideas in~\cite{SlyReconColouring}.

Let us be more specific. For the probability measure $\mu$ we consider, let  $R_r$ be the {\em ratio of  Gibbs marginals}  at the root $r$
defined by 
\begin{equation}\label{eq:DefOfRf}
R_r=\frac{\mu_r(+1)}{\mu_r(-1)} \enspace.
\end{equation}
Recall that  $\mu_r(\cdot)$ denotes the marginal of the Gibbs distribution $\mu(\cdot)$ 
at the root $r$. 

For a vertex $u\in V$,  we let $T_u$ be the subtree of $T$ that includes $u$, and all its descendants. 
Also, we let  $R_u$ be the ratio of marginals at vertex $u$, where the Gibbs distribution is, now,  
with respect to the subtree $T_u$. 

Suppose that the vertices $w_1, \ldots, w_{\Delta}$ are the children of the root $r$.  Our focus is on 
expressing $\log R_r $ recursively, as a function of $\log R_{w_1}, \ldots, \log R_{w_{\Delta}}$. 
Note that we study the logarithm of the ratios involved, which can be viewed as applying the 
potential function  $\log (\cdot)$ to the tree recursions. 
We have that  $\log\left(R_r\right)=H\left(\log R_{w_1},   \ldots, \log R_{w_{\Delta}}\right)$ where 
\begin{equation}\label{def:OfHOverview}
H(x_1,x_2, \ldots, x_{\Delta}) =\sum^{\Delta}_{i=1} \log\left( \frac{\exp \left(x_i +\beta J_{\{ r,w_i\}} \right)+1}{\exp(x_i)+\exp\left(\beta J_{\{r,w_i\}} \right)} \right) \enspace.
\end{equation}
Note that  $J_{\{r,w_i\}}$ is the coupling parameter that corresponds to the edge between the root $r$ with its child $w_i$.

All the above extends naturally in the case where we impose boundary conditions. That is, for a region $K\subseteq V$,
and $\tau\in \{\pm 1\}^K$, we define the ratio of marginals $R^{K,\tau}_r$ at the root, where now the ratio
is between the conditional marginals $\mu_r(+1 \ |\ K,\tau)$ and $\mu_r(-1 \ |\ K,\tau)$. The recursive function $H$ for the conditional
ratios is exactly the same as the one above. 

Our interest is on the {\em gradient} of the function $H$. Specifically, for every $i\in [\Delta]$, we let
\begin{equation}\label{eq:InfDef}
\Inf_{\{ r,w_i\}} =  \sup_{x_1,\ldots, x_\Delta} \left| \frac{\partial }{\partial x_i} H(x_1,x_2, \ldots, x_{\Delta}) \right| \enspace.
\end{equation}
It turns out that, in our case,  $\Inf_{\{ r,w_i\}} $ has a simple form
\begin{equation*}
\Inf_{\{ r,w_i\} } =  \frac{\left|1-\exp\left( \beta J_{\{ r,w\} } \right) \right|}{1+\exp\left( \beta J_{\{ r,w_i\}}\right) }\enspace.
\end{equation*}
Utilising the idea of down-up coupling  from  \cite{BhatVeraVigodaWeitz}, we prove the following: 
\begin{equation}\label{eq:ApproachDetBound}
 \lnorm \mu_h(\cdot \ |\ \bsigma(r)=+ 1) -\mu_h(\cdot \ |\ \bsigma(r)=-1 )  \rnorm_{\rm TV}  \le
\sqrt{\sum_{v\in \Lambda}\;\prod_{e\in\mathrm{path}(r,v)} \Inf_e^2 } \enspace,
\end{equation}
where $\Lambda=\Lambda(h)$ denotes the set of vertices at distance $h$ from the root $r$. 
Note that the above provides a bound for the total variation distance of the the marginals for fixed, i.e., non-random, 
couplings $\{J_{e}\}_{e\in E}$. 
Inequality \eqref{eq:ApproachDetBound}, extends naturally when we study reconstruction for the distribution $\bmu$ defined in \eqref{def:SpinGlass}, i.e.,
when the coupling parameters $\bJ_e$ are i.i.d. samples from a distribution $\phi$.  
Indeed, averaging yields 
\begin{equation}\label{eq:HighLevel4SecondMomentNonRecon}
\mathbb{E}\left[ \left( \lnorm \bmu_h(\cdot \ |\ \bsigma(r)=+ 1) -\bmu_h(\cdot \ |\ \bsigma(r)=-1 )  \rnorm_{\rm TV} \right)^2\right] \le
\sum_{v\in \Lambda}\;\prod_{e\in\mathrm{path}(r,v)} \mathbb{E} \left[\mathbold{\Gamma}_e^2\right] \enspace, 
\end{equation}
where we have $\mathbold{\Gamma}_{e}=\frac{\left|1-\exp\left( \beta \bJ_{e } \right) \right|}{1+\exp\left( \beta \bJ_{e}\right)}$, for each $e\in E$. 
Note that the above holds, since each $\mathbold{\Gamma}_e$ depends only on $\bJ_e$, while the coupling parameters 
$\bJ_e$ are assumed to be independent with each other.  

At this point, and since the $\bJ_e$'s are identically distributed, we further observe  that  for any $e\in E$,  we have that
\begin{equation*}
\Delta_{\rm KS}(\beta,\phi)=\left(\mathbb{E}\left[\mathbold{\Gamma}_e^2\right] \right)^{-1}\enspace. 
\end{equation*}
Since the underlying tree  $T$ is $\Delta$-ary, it is immediate to see that for  
$\Delta<\Delta_{\rm KS}(\beta, \phi)$, the r.h.s. of~\eqref{eq:HighLevel4SecondMomentNonRecon} tends to zero 
as $h\to\infty$.  
From this point on,  it is standard to prove non-reconstruction.

Our analysis   allows to deal with the randomness of the spin-glass measure $\bmu$  by utilising 
the bound in~\eqref{eq:ApproachDetBound}. That is, the upper bound on the total variation distance 
has a nice {\em product} form of the quantities $\Gamma_e$, which,  in turn,   expresses the dependence 
of the total variation distance on the  edge couplings $\{J_e\}_{e\in E}$. 
This product form of the bound, behaves rather nicely when we need to take averages over the 
randomness of the coupling parameters $\{\bJ_e\}_{e\in E}$ of the the spin-glass measure $\bmu$.

\subsection*{Reconstruction}

In the reconstruction regime, the configuration at the root has a non-vanishing effect on the configuration
of the vertices at distance $h$, regardless of the height $h$. 
Specifically, the  corresponding  leaf configurations from the measure conditioned on root's spin being $+1$, and $-1$, 
  are so different  with each other,  that any discrepancies  cannot be attributed to random fluctuations.
Therefore, a question that naturally arises is how can we take advantage of the discrepancies
so that we  infer the spin of the root.

For the standard  ferromagnetic Ising, several approaches  have been developed to establish
reconstruction (see \cite{EvKenPerSchulConConIsingTree}, \cite{BorgsCMR06}, \cite{ioffe1996extremality}). 
Here, we build on an elegant argument in \cite{EvKenPerSchulConConIsingTree}.
The authors in this work, show that
a simple {\em majority vote} of the leaf spins, conveys information 
sufficient  to reconstruct root's spin,  The  majority vote on the leaves is defined by
\begin{equation}\label{def:MajVote}
M_{h}=\sum_{u\in \Lambda}\sigma(u) \enspace. 
\end{equation}
The estimation rule is to infer that the spin at the root is  ${\rm sgn}\{M_{h}\}$, 
i.e., the sign of $M_h$. Impressively, it turns out that this estimator  is optimal, i.e., it  coincides with the {\em maximum 
likelihood} one. For  the $\Delta$-ary tree, one establishes  reconstruction for the ferromagnetic  Ising model  by employing  
a second moment  argument on the estimator $M_h$.  

For the distributions we consider here, the above estimator  is far from sufficient. This is due to various facts.
Firstly,  we allow for mixed couplings on the edges, i.e., certain edges can be ferromagnetic, and others can be anti-ferromagnetic. 
Secondly, the strength of the interaction, i.e., the magnitude of $\J_e$'s,  is expected to vary from one edge to the other. 
To this end, we introduce a new estimator, and we establish reconstruction by building on the second moment argument from  
\cite{EvKenPerSchulConConIsingTree}.  The starting point towards deriving this estimator,  comes from just considering the standard anti-antiferromagnetic Ising. 
The statistic from \eqref{def:MajVote}, clearly does not  work for  this distribution. However, there is an easy  remedy,
by   taking into account the parity of  the height $h$, i.e., if $h$ is an even, or an odd number. 
We  infer that the spin at the root is equal to   ${\rm sgn}\left\{\widehat{M}_{h}\right\}$,  where
\begin{equation*}
\widehat{M}_{h}=(-1)^{h} \sum_{u\in \Lambda}\sigma(u) \enspace. 
\end{equation*}
For the spin-glass distributions we consider here, we need to get  the above idea even further. Firstly,  
in order to accommodate the   {\em mixed} ferromagnetic and anti-ferromagnetic couplings on the edges of the tree. 
It seems meaningful to use  the estimator ${\rm sgn}\left\{\widetilde{M}_{h} \right\}$ for the root configuration,  where
\begin{equation*}
\widetilde{M}_{h}=\sum_{u\in \Lambda}  \sigma(u) \prod_{e\in {\rm path}(r,u)} {\rm sign}\{\bJ_{e}\} \enspace,
\end{equation*}
with ${\rm path}(r,u)$ denoting the set of edges along the unique path connecting $r$ to $u$.
So that in $\widetilde{M}_{h}$,  for each leaf we essentially examine the parity of the number of antiferromagnetic couplings along the path that 
connects it to the root. Unfortunately, for the above estimator, our  second moment argument does not seem to work all that well. 

The estimator we end up using, is a  {\em reweighted} version of $\widetilde{M}_{h}$, which we call
the  ``flip majority" vote, and is defined by 
\begin{equation*}
{F}_{h}=\sum_{u\in \Lambda}  \sigma(u) \prod_{e\in {\rm path}(r,u)} 
{\textstyle \frac{ 1-\exp\left( \beta \bJ_e \right)}{1+\exp\left( \beta \bJ_e \right)}} \enspace. 
\end{equation*}
Note that the absolute value of the weight for the edge $e$, above, coincides with the quantity $\mathbold{\Gamma}_e$
in \eqref{eq:HighLevel4SecondMomentNonRecon}.
Naturally, the estimation rule is to infer that the root spin is  ${\rm sgn}\left\{F_{h}\right\}$. 

\section{Tree recursions and Influences}\label{sec:TreeRecurInf}

What follows applies to any kind of tree. For the sake of simplicity, in this section, 
we consider the $\Delta$--ary tree $T=(V,E)$ rooted at $r$.  
Suppose  that we are given the number  $\beta\geq 0$, while each  edge $e\in E$ has its 
own coupling parameter,~$J_e$.
Assume, for the moment,  that the  coupling parameters at the edges are fixed, i.e., they are arbitrary real numbers. 
Given $\beta$ and $\{J_e\}_{e\in E}$,  we consider the Gibbs distribution $\mu=\mu_{\beta, \{J_e\}}$ 
similarly to the one we have  in \eqref{def:SpinGlass}.  That is, every  $\sigma\in\{\pm 1\}^V$ gets 
a probability mass defined by
\begin{align}\label{def:GibbDistrTreeRecur}
\mu(\sigma) &\propto \textstyle \exp\left( \beta \cdot \sum_{\{w,u\}\in E} {\bf 1}\{\sigma(u)=\sigma(w)\} \cdot J_{\{u,w\}}   \right) \enspace.
\end{align}
For a region  $K \subseteq V \setminus\{r\}$ and   $\tau\in  \{\pm 1\}^{K}$, we consider the {\em ratio of 
marginals} at the root $R^{K, \tau}_r$  such that
\begin{align}\label{eq:DefOfR}
R^{K, \tau}_r=\frac{\mu_r(+1\ |\ K,  \tau )}{\mu_r(-1\ |\ K,  \tau)} \enspace.
\end{align}
Recall that  $\mu_r(\cdot \ |\ K,  \tau )$ denotes the marginal of the Gibbs distribution $\mu(\cdot \ |\ K,  \tau )$ 
at the root $r$. Also, note that the  above  allows  for $R^{K, \tau}_r=\infty$, when 
$\mu_r(-1\ |\ K,  \tau)=0$.

For a vertex $u\in V$,  we let $T_u$ be the subtree of $T$ that includes $u$, and all its descendants. 
We always assume  that the root of $T_u$ is the vertex $u$.  
With a slight abuse of notation,  we let  $R^{K,\tau}_u$ denote the ratio of marginals at the root for 
the subtree $T_u$, where the Gibbs distribution is, now,  with respect to $T_u$.

Suppose that the root $r$ is of degree $\Delta>0$, while let the vertices $w_1, \ldots, w_{\Delta}$
be its children.  We express $R^{K, \tau}_r$ it terms of  $R^{K, \tau}_{w_i}$'s 
by having  $R^{K, \tau}_{r}=F_\Delta(R^{K, \tau}_{w_1}, R^{K, \tau}_{w_2}, 
\ldots, R^{K, \tau}_{w_\Delta})$,   for  
\begin{align*}
F_\Delta:[0, +\infty]^\Delta\to [0, +\infty] & &\textrm{such that}& &
 (x_1, \ldots, x_\Delta)\mapsto \prod^\Delta_{i=1}\frac{\exp(\beta J_{\{r, w_i\}}) {x}_i+1}{{x}_i+\exp(\beta J_{\{r, w_i\}})} \enspace.
\end{align*}

For  the analysis that  follows, we get cleaner results by equivalently working with log-ratios rather 
than  ratios of Gibbs marginals.  Let  $H_\Delta=\log \circ F_\Delta \circ \exp$,  which means that $H_\Delta:[-\infty, +\infty]^\Delta\to [-\infty, +\infty] $ is such that
\begin{align*}
 \textstyle (x_1, \ldots, x_\Delta)\mapsto \sum^\Delta_{i=1}
\log\left( \frac{\exp(\beta J_{\{r, w_i\}}+x_i)+1}{\exp(x_i)+\exp(\beta J_{\{r, w_i\}})} \right) \enspace.
\end{align*}
From  the above, it is elementary to verify that   $\log R^{K, \tau}_{r}=H_\Delta(\log R^{K, \tau}_{w_1}, 
\ldots, \log R^{K, \tau}_{w_\Delta})$. The above transformation is standard in the literature, and can be viewed as applying the potential function $\log(\cdot)$ in 
the tree recursion. 
For every $i\in [\Delta]$, we let
\begin{align}\label{def:OfGammae}
\Inf_{\{ r,w_i\}} &= \sup_{x_1,\ldots, x_{\Delta}} \left|\frac{\partial }{\partial x_i} H(x_1,x_2, \ldots, x_{\Delta}) \right|
\enspace.
\end{align}
The quantities $\{\Gamma_e\}_{e\in E}$ arise naturally in various settings in  our 
analysis. Specifically, we use the theorem below, which follows as corollary from results in \cite{OptMCMCIS,VigodaSpectralInd}. 

\begin{theorem}\label{thrm:PairwiseInfluenceBound}
For $\beta>0$, consider the tree $T=(V, E)$ and $\{J_e\}_{e\in E}$ for fixed $J_{e}\in \mathbb{R}$.
Let the Gibbs distribution $\mu=\mu_{\beta, \{J_e\}}$ on $T$,  defined as in \eqref{def:GibbDistrTreeRecur}.

For any two vertices $u,w\in V$, for any $M \subseteq V\setminus\{u,w\}$, and any $\tau\in \{\pm 1\}^{M}$
the following holds:
\begin{align}
\lnorm \mu_{w}(\cdot\ |\ (M, \tau), (u,+1))-\mu_{w}(\cdot\ |\ (M, \tau), (u,-1)) \rnorm_{\rm TV} &
\leq \prod_{e\in\mathrm{path}(u,w)} \Gamma_e \enspace,
\end{align}
where $\mathrm{path}(u,w)$  is the set of edges along the path from $u$ to $w$ in $T$, 
while  $\Gamma_e$'s are defined   in \eqref{def:OfGammae}.
\end{theorem}
Specifically, Theorem  \ref{thrm:PairwiseInfluenceBound} is a direct consequence of Lemma B.2 in \cite{OptMCMCIS},
and Lemma 15 in \cite{VigodaSpectralInd}.
For the distributions we consider in this work, it turns out,  that the quantities  $\Inf_{\{ r,w_i\}} $ have a simple form
which, somehow,  is a reminiscent of the quantity $\Delta_{\rm KS}$ in \eqref{eq:DeltaKSBPhi}.
\begin{claim}\label{cl:WhatIsGamma}
For $e=\{r,w_i\} \in E$, consider the quantity $\Inf_e$ defined in \eqref{def:OfGammae}. 
We have that  
\begin{align}\label{eq:whatisgamma}
\Inf_{e} =  \frac{\left|1-\exp\left( \beta J_{e} \right) \right|}{1+\exp\left( \beta J_{e}\right) } \enspace.
\end{align}
\end{claim}

\begin{proof}[Proof of Claim~\ref{cl:WhatIsGamma}]

The derivations below are standard and we present them for the sake of our work being self-contained. 
For $i\in [\Delta]$ and $e=\{r,w_i\}$, let $h_i:[-\infty,+\infty]\to \mathbb{R}$ be the function  

\begin{align}\nonumber 
x \mapsto -\frac{(1-\exp(2\beta J_e))\cdot \exp(x)}{(\exp(\beta J_{e}+x)+1)(\exp(x)+\exp(\beta J_e))} \enspace.
\end{align}
It is easy to verify that $\frac{\partial }{\partial x_i} H_\Delta(x_1,  \ldots, x_\Delta)=h_i(x_i)$. It is also straightforward to see that for any real function $f$ we have 
\begin{equation*}
\sup_{x}|f(x)| = \textstyle
\max\left\{ \left|\sup_{x}f(x)\right|, \left|\inf_{x}f(x)\right|\right\} \enspace,
\end{equation*}
so that
\begin{equation}\label{eq:SimpleInf}
	\Inf_{e} = 	
					\sup_{x_1,\ldots, x_d} \left|\frac{\partial }{\partial x_i} H(x_1,x_2, \ldots, x_{d})\right| 
				 =	\sup_{x}|h_i(x)|
				 = 	\max\left\{ \left|\sup_{x}h_i(x)\right|, \left|\inf_{x}h_i(x)\right|\right\} 
				 \enspace.
\end{equation}
Now let also $b_i=\exp(\beta J_{e}) > 0$, so that 
\begin{equation}\nonumber
	h_i(x) =
	-\frac{{\mathrm e}^{x} \left(1-b_i^{2}\right)}{\left({\mathrm e}^{x} b_i+1\right) \left({\mathrm e}^{x}+b_i\right)} 
	\enspace,  
\end{equation}
and notice that we want to show 
\begin{equation*} 
	\Inf_{e } 
		=  \frac{\left|1-b_i\right|}{1+b_i }
	\enspace. 
\end{equation*}
First, if $b_i=1$, then \eqref{eq:SimpleInf} gives 
\begin{equation}\nonumber 
\Inf_{e } = \sup_{x}|h_i(x)| = 0 = \frac{\left|1-b_i\right|}{1+b_i } \enspace.
\end{equation}
 Assume now $b_i \neq 1$. Differentiating $h_i$ gives 
\begin{equation}\nonumber 
h^{\prime}_i(x) =-\frac{{\mathrm e}^{x} \left(b_i^{2}-1\right) b_i \left({\mathrm e}^{2 x}-1\right)}{\left({\mathrm e}^{x} b_i+1\right)^{2} \left({\mathrm e}^{x}+b_i\right)^{2}}\enspace.
\end{equation}
Since $b_i>0$, and $b_i \neq 1$, we observe that $h^{\prime}_i$ vanishes only at $x=0$, and in particular, $x=0$ must be the only sign alternation point of $h^{\prime}_i$. Finally, it is elementary to check that 
\begin{equation}\nonumber 
\lim_{x\to \infty} h_i(x) = \lim_{x\to -\infty} h_i(x) = 0 \enspace.
\end{equation}
Therefore, $0$ and $h_i(0)$ must be the global optima of $h_i$. Hence, \eqref{eq:SimpleInf} yields

\begin{equation*}
	\Inf_{e} 
		= 	\max\left\{ \left|\sup_{x}h_i(x)\right|, \ \left|\inf_{x}h_i(x)\right|\right\} 
		= 	|h_i(0)| 
		=	\frac{\left|1-b_i^{2}\right|}{\left( b_i+1\right) \left(1+b_i\right)}
		=	\frac{\left|1-b_i\right|}{ 1+b_i} \enspace,
\end{equation*}
as desired. 
\end{proof}

\section{Theorem \ref{thrm:GeneralDeltary} - Proof of non-reconstruction.}\label{sec:thrm:GeneralDeltary}

In order to prove Theorem \ref{thrm:GeneralDeltary}, first consider the distribution we define in \eqref{def:GibbDistrTreeRecur}, 
in Section \ref{sec:TreeRecurInf}. That is,  for a  tree $T=(V,E)$ rooted at $r$,  assume that we are given the  parameters 
$\beta>0$  and $\{J_e\}_{e\in E}$, such that $J_e\in \mathbb{R}$.  Note that $J_e$ are fixed real constants, i.e.,  they are not random numbers. 

We define the Gibbs distribution $\mu=\mu_{\beta, \{J_e\}}$ on the tree $T$ such that  each $\sigma \in \{\pm 1\}^V$   is assigned probability 
measure $\mu(\sigma)$ such that 
\begin{align}\label{def:NonUniform}
\mu(\sigma) &\propto \textstyle \exp\left( \beta\sum_{\{w,u\}\in E} {\bf 1}\{\sigma(u)=\sigma(w)\} \cdot J_{\{u,w\}}   \right) \enspace .
\end{align}
For two vertices $u,w$ in $T$, write $\mathrm{path}(u,w)$ for the set of edges in the unique path from $u$ to $w$.
Building on Theorem \ref{thrm:PairwiseInfluenceBound},  for the aforementioned distribution we have the following result:

\begin{theorem}\label{thrm:ReconBoundNonUniform} 
For integer $h>0$,  $\beta>0$, and $\{J_e\}_{e\in E}$ such that $J_e\in \mathbb{R}$, 
let $T=(V,E)$ be an arbitrary tree of height $h$, rooted at vertex $r$, and let   
the Gibbs distribution $\mu=\mu_{\beta, \{J_e\}}$ on $T$ be defined as in \eqref{def:NonUniform}.

We have that
\begin{align}\label{eq:thrm:ReconBoundNonUniform}
 \lnorm \mu_h(\cdot \ |\ \bsigma(r)=+ 1) -\mu_h(\cdot \ |\ \bsigma(r)=-1 )  \rnorm_{\rm TV}  &\le
\sqrt{\sum_{v\in \Lambda}\; \prod_{e\in\mathrm{path}(r,v)} \Inf_e^2 } \enspace,
\end{align}
where $\Inf_e$ is the influence of edge $e$ defined in \eqref{eq:whatisgamma}, and 
 $\Lambda$ is  the set of vertices at distance $h$ from the root.  
\end{theorem}

For the above, recall that $\mu_{h}$ is the marginal of $\mu$ on the set of vertices at distance $h$ from the root, i.e., the set $\Lambda$. 
In light of  Theorem~\ref{thrm:ReconBoundNonUniform},  the non-reconstruction part of Theorem~\ref{thrm:GeneralDeltary}  follows as a corollary.

\begin{proof}[Proof of Theorem \ref{thrm:GeneralDeltary} - Non-Reconstruction]
Consider the Gibbs distribution $\bmu=\bmu_{\beta, \phi}$ on the $\Delta$-ary tree $T=(V,E)$, and let
 $\Delta_{\rm KS}=\Delta_{\rm KS}(\beta,\phi)$ be defined as in  \eqref{eq:DeltaKSBPhi}.
We need to show that for $\Delta<\Delta_{\rm KS}$ we have  
\begin{equation}\label{eq:Target4NonReconThrm:GeneralDeltary}
\limsup_{h \to \infty}
\Exp\left[  \lnorm \bmu_h(\cdot \ |\ \bsigma(r)=+ 1) -\bmu_h(\cdot\ |\ \bsigma(r)=-1)  \rnorm_{\rm TV} \right] =0 \enspace.
\end{equation}
Given the $\sigma$-algebra generated 
by the coupling  parameters $\{\bJ_e\}_{e\in E}$, from Theorem \ref{thrm:ReconBoundNonUniform}, we have  that 
\begin{align}\label{eq:Base4thrm:GeneralDeltary}
\left( \lnorm \bmu_h(\cdot \ |\ \bsigma(r)=+ 1) -\bmu_h(\cdot\ |\ \bsigma(r)=-1)  \rnorm_{\rm TV} \right)^2 &\leq \sum_{v\in \Lambda}\;\prod_{e\in\mathrm{path}(r,v)} \mathbold{\Inf}_e^2\enspace, 
\end{align}
where recall that $\Lambda$ is the set of vertices at distance $h$ from the root $r$, while for every $e\in E$ we have that 
\[ \textstyle\mathbold{\Inf}_{e} =  \frac{\left|1-\exp\left( \beta \bJ_{e } \right) \right|}{1+\exp\left( \beta \bJ_{e}\right) }\enspace.\]
For the sake of brevity, we let 
\begin{align}
\bmu^+(\cdot )&= \bmu(\cdot \ |\ \bsigma(r)=+ 1)\enspace
, & \textrm{and} && \bmu^-(\cdot )&= \bmu(\cdot \ |\ \bsigma(r)=-1 )\enspace. 
\end{align}
Then, from \eqref{eq:Base4thrm:GeneralDeltary} we have that
\begin{equation}\label{eq:ExpQSqareBound}
\Exp \left[ \left( \lnorm \bmu^+_h(\cdot ) -\bmu^-_h(\cdot)  \rnorm_{\rm TV} \right)^2\right]
\le \sum_{v\in \Lambda }\;\prod_{e\in\mathrm{path}(r,v)} \Exp \left[ \mathbold{\Inf}_e^2\right]  \enspace,
\end{equation}
where the expectation is with respect to random variable $\bJ_e$. 
We derive the r.h.s. of the equation above  using the observation that  each $\mathbold{\Gamma}_e$ depends only on $\bJ_e$, and 
the coupling parameters $\{\bJ_e\}$, are assumed to be {\em independent} with each other. 

Furthermore, our assumption that $\Delta<\Delta_{\rm KS}$, corresponds to having that 
$\Exp \left[ \mathbold{\Inf}_e^2\right]<\Delta^{-1}$.  Hence, there exists $\varepsilon\in (0,1]$  such that 
\begin{align*}
\Exp \left[ \mathbold{\Inf}_e^2\right]  = \frac{1-\varepsilon}{\Delta}\enspace .
\end{align*}
Using the above, \eqref{eq:ExpQSqareBound}, and the fact that $T$ is $\Delta$-ary, and hence, the size of $\Lambda$ is $\Delta^{h}$, we get that
\begin{align*}
\Exp \left[ \left( \lnorm \bmu^+_h(\cdot ) -\bmu^-_h(\cdot)  \rnorm_{\rm TV} \right)^2 \right] \leq (1-\varepsilon)^{h}\enspace. 
\end{align*}

\noindent
Invoking Markov's inequality we further get that
\begin{equation*}
\Pr\left[ \left( \lnorm \bmu^+_h(\cdot ) -\bmu^-_h(\cdot)  \rnorm_{\rm TV} \right)^2
\ge (1-\varepsilon)^{h/2}\right] \le (1-\varepsilon)^{h/2} \enspace,
\end{equation*}
or, since $\lnorm \bmu^+_h(\cdot ) -\bmu^-_h(\cdot)  \rnorm_{\rm TV}  \geq 0$, we equivalently have that  
\begin{equation*}
\Pr\left[  \lnorm \bmu^+_h(\cdot ) -\bmu^-_h(\cdot)  \rnorm_{\rm TV}
\ge (1-\varepsilon)^{h/4}\right] \le (1-\varepsilon)^{h/2}  \enspace.
\end{equation*}
Furthermore, since $\ \lnorm \bmu^+_h(\cdot ) -\bmu^-_h(\cdot)  \rnorm_{\rm TV}  \le 1$ and $\Pr\left[ \lnorm \bmu^+_h(\cdot ) -\bmu^-_h(\cdot)  \rnorm_{\rm TV} < (1-\varepsilon)^{h/4}\right] \leq 1$,  we have
\begin{align*}
\lefteqn{
\Exp\left[  \lnorm \bmu^+_h(\cdot ) -\bmu^-_h(\cdot)  \rnorm_{\rm TV} \right] } 
\vspace{1cm}\\
& \leq (1-\varepsilon)^{h/4}
\Pr\left[\lnorm \bmu^+_h(\cdot ) -\bmu^-_h(\cdot)  \rnorm_{\rm TV}< (1-\varepsilon)^{h/4}\right]
+ 
\Pr\left[\lnorm \bmu^+_h(\cdot ) -\bmu^-_h(\cdot)  \rnorm_{\rm TV}\ge (1-\varepsilon)^{h/4}\right]
\\
&\le  (1-\varepsilon)^{{h}/{4}} +(1-\varepsilon)^{{h}/{2}} \le 2(1-\varepsilon)^{{h}/{4}}.
\end{align*}
The above implies \eqref{eq:Target4NonReconThrm:GeneralDeltary}, and concludes the non-reconstruction part of
Theorem \ref{thrm:GeneralDeltary}.
\end{proof}

\section{Proof of Theorem \ref{thrm:ReconBoundNonUniform}}\label{sec:thrm:ReconBoundNonUniform}

Recall that we are dealing with the Gibbs distribution
$\mu=\mu_{\beta, \{J_e\}}$ on the tree $T$ of height $h$.
With respect to $\mu$ and every edge $e$ of the tree, we obtain the influence $\Inf_e$ in the standard way. 

To prove Theorem ~\ref{thrm:ReconBoundNonUniform}, we use the idea of down-up coupling from \cite{BhatVeraVigodaWeitz}, which also relies on ideas in \cite{SlyReconColouring}. 
To this end,  let us introduce a few notions. For $s\in \{\pm1\}$, we let $\mdownup^{s}$ be the distribution on the configuration at the root $r$ of the tree $T=(V,E)$ that is induced by
the following experiment.  Recall that $\Lambda$ is the set of vertices at distance $h$ from the root.  First, we obtain the configuration
$\bsigma\in \{\pm 1\}^V$ on the tree from the measure  $\mu^s(\cdot )$, where 
\begin{align*}
\mu^s(\cdot )&= \mu(\cdot \ |\ \bsigma(r)=s) \enspace. 
\end{align*}
Next, we erase all the assignments apart from those at the vertices in $\Lambda$. Then, we obtain a new configuration, $\btau$, from the distribution  $\mu^{\Lambda, \bsigma}$, i.e., the distribution $\mu$ conditional on the configuration 
of set $\Lambda$ be as in $\bsigma$. With the measure $\mdownup^{s}$ we denote the distribution of $\btau(r)$, i.e., the assingment
of $\btau$ at the root $r$. 

Recall now that for $s \in \{\pm1\}$, we write  $\mu^{s}_{h}(\cdot)$ for the marginal of $\mu$ at the vertices at distance $h$ from the root, conditioned on $\bsigma(r) = s$.  The following lemma  was essentially proved for standard Gibbs distributions in \cite{BhatVeraVigodaWeitz}. For the sake of completeness, we present our own proof for the spin-glasses in the Appendix \ref{sec:lem:UpDwnUpperBound}.

\begin{lemma}[\cite{BhatVeraVigodaWeitz} ]\label{lem:UpDwnUpperBound}
For integer $h>0$, let $T=(V,E)$ be an arbitrary tree of height $h$, rooted at vertex $r$.  
For any $\beta>0$, for any $\{J_e\}_{e\in E}$ with $J_e\in \mathbb{R}$,   let  
the Gibbs distribution $\mu=\mu_{\beta, \{J_e\}}$ on $T$ be defined as in~\eqref{def:NonUniform}.  Then,
\begin{align}\label{eq:lem:UpDwnUpperBound}
 \lnorm \mu^{+}_{h}(\cdot) -\mu_{h}^{-}(\cdot)  \rnorm_{\rm TV}  &\leq
 \sqrt{\lnorm\mdownup^{+}(\cdot)-\mdownup^{-}(\cdot)\rnorm_{\rm TV} } \enspace.
\end{align}
\end{lemma}
We prove the upper bound in \eqref{eq:thrm:ReconBoundNonUniform} be means of the Lemma \ref{lem:UpDwnUpperBound}, i.e., 
by bounding appropriately the quantity on the r.h.s. of \eqref{eq:lem:UpDwnUpperBound}.  Specifically, we use the bound obtained in the following
proposition.

\begin{proposition}\label{prop:UpDownDistance}
For integer $h>0$, let $T=(V,E)$ be an arbitrary tree of height $h$, rooted at vertex $r$, and write $\Lambda$ for the set of vertices
at distance $h$ from the root. For  each $e \in E$, let  $\Inf_e$  be the influence of edge $e$, given by \eqref{eq:whatisgamma}. Then,
\begin{align} \label{eq:prop:UpDownDistance}
\lnorm\mdownup^{+}(\cdot)-\mdownup^{-}(\cdot)\rnorm_{\rm TV}& \leq 
\sum_{v\in \Lambda}\;\prod_{e\in\mathrm{path}(r,v)} \Inf_e^2 \enspace,
\end{align}
where $\mdownup^{+}$, $\mdownup^{-}$ are as in Lemma~\ref{lem:UpDwnUpperBound}.

\end{proposition}
The proof of Proposition \ref{prop:UpDownDistance} appears in Section \ref{sec:prop:UpDownDistance}. Now, Theorem \ref{thrm:ReconBoundNonUniform} follows by plugging \eqref{eq:prop:UpDownDistance} into
\eqref{eq:lem:UpDwnUpperBound}, i.e., we have that 
\begin{align*}
\lnorm \mu^{+}_h(\cdot) -\mu^{-}_h(\cdot) \rnorm_{\rm TV}  &\leq
 \sqrt{\lnorm\mdownup^{+}(\cdot)-\mdownup^{-}(\cdot)\rnorm_{\rm TV} }\le \sqrt{
\sum_{v\in \Lambda}\;\prod_{e\in\mathrm{path}(r,v)} \Inf_e^2}\enspace.
\end{align*}

\section{Proof of Proposition \ref{prop:UpDownDistance}}\label{sec:prop:UpDownDistance}

For $s \in \{\pm 1\}$, recall that  $\mu^{s}_{\Lambda}(\cdot)$ be the marginal of $\mu$ on  $\Lambda$,  conditional on the configuration at $r$ being $s$. For any configuration $\tau \in \{\pm 1\}^\Lambda$,  also recall that $\mu^{\Lambda,\tau}_r$ 
is the marginal of $\mu$ at the root $r$, conditional on the configuration at $\Lambda$  being $\tau$. 
 In order to prove Proposition~\ref{prop:UpDownDistance} we use the following result. 

\begin{lemma}\label{lem:NoClue} 
For integer $h>0$, let $T=(V,E)$ be an arbitrary tree of height $h$ rooted at vertex $r$, 
and write $\Lambda$ for the vertices at distance $h$ from the root.  
For any $\beta>0$, for any $\{J_e\}_{e\in E}$ such that $J_e\in \mathbb{R}$,   let  
the Gibbs distribution $\mu=\mu_{\beta, \{J_e\}}$ on $T$ be defined as in \eqref{def:NonUniform}.

Then, for any distribution $\nu:\{\pm 1\}^{\Lambda}\times \{\pm 1\}^{\Lambda}$, coupling of the marginals 
$\mu^{+}_{\Lambda}(\cdot)$  and $\mu^{-}_{\Lambda}(\cdot)$, we have 
\begin{align}
\lnorm\mdownup^{+}(\cdot)-\mdownup^{-}(\cdot)\rnorm_{\rm TV} &\le 
\Exp_{(\bsigma, \btau ) \sim \nu} \left[
	\left\| \mu^{\Lambda, \bsigma}_r (\cdot) - \mu^{\Lambda,\btau}_r (\cdot )\right\|_{\rm TV}
	\right]\enspace.
\end{align}
\end{lemma}
\begin{proof}
We have that 
\begin{align}
\lefteqn{
\Exp_{(\bsigma, \btau ) \sim \nu} \left[
\left\| \mu^{\Lambda, \bsigma}_r (\cdot) - \mu^{\Lambda,\btau}_r (\cdot )\right\|_{\rm TV}
\right]
} \hspace{3cm} \nonumber \\
&= \sum_{\sigma,\tau\in\{\pm1\}^{\Lambda} } 
	\left|  \mu^{\Lambda, \sigma}_r (+1) - \mu^{\Lambda,\tau}_r (+ 1)\right| \cdot\nu(\sigma,\tau ) \nonumber \\
&\ge \left|\sum_{\sigma , \tau \in\{\pm1\}^{\Lambda} } 
\left( \mu^{\Lambda, \sigma}_r (+1) - \mu^{\Lambda,\tau}_r (+ 1) \right) \cdot\nu(\sigma,\tau )\right| \nonumber \\
&= \left|\left(\sum_{\sigma \in\{\pm1\}^{\Lambda} } 
\mu^{\Lambda, \sigma}_r(+1)\cdot \mu^+_{\Lambda}(\sigma ) \right)
- 
\left(\sum_{\tau \in\{\pm1\}^{\Lambda} } 
\mu^{\Lambda, \tau}_r(+1) \mu^-_{\Lambda}(\tau ) \right)\right| \enspace. \label{eq:BasicIneqlem:NoClue}
\end{align}
The second derivation  is due the triangle inequality, while last equality holds since $\nu(\cdot, \cdot)$ is a coupling of 
$\mu^{+}_{\Lambda}(\cdot)$  and $\mu^{-}_{\Lambda}(\cdot)$, and thus for any $\sigma\in \{\pm 1\}^{\Lambda}$ 
we have that 
\begin{align*}
\sum_{\tau\in \{\pm 1 \}^{\Lambda}}\nu(\sigma,\tau) &= \mu^{+}_{\Lambda}(\sigma)\enspace, & \textrm{and} &&
\sum_{\tau\in \{\pm 1 \}^{\Lambda}}\nu(\tau, \sigma) &= \mu^{-}_{\Lambda}(\sigma) \enspace. 
\end{align*}
Furthermore, note that for any $s,t\in \{\pm 1\}$, we have that
\begin{align}\label{eq:downupis}
\mu^s_{\downarrow \uparrow}(t) &=\sum_{\sigma \in\{\pm\}^{\Lambda} } 
\mu^{\Lambda, \sigma}_r(t)\cdot \mu_{\Lambda}^{s}(\sigma ) \enspace. 
\end{align}
The above follows from the definition of $\mu^s_{\downarrow \uparrow}$, and the law of total probability.
Plugging \eqref{eq:downupis} into \eqref{eq:BasicIneqlem:NoClue}, we get that 
\begin{align*}
\Exp_{(\bsigma, \btau ) \sim \nu} \left[
\left\| \mu^{\Lambda, \bsigma}_r (\cdot) - \mu^{\Lambda,\btau}_r (\cdot )\right\|_{\rm TV}
\right]
&\geq  
\left|\mu^+_{\downarrow \uparrow}(+) 
- \mu^{-}_{\downarrow \uparrow}(+)\right| 
 = \|\mu^+_{\downarrow \uparrow}(\cdot) 
- \mu^{-}_{\downarrow \uparrow}(\cdot)\|_{\rm TV} \enspace.
\end{align*}
The above concludes the proof of Lemma \ref{lem:NoClue}.
\end{proof}

Lemma \ref{lem:NoClue} implies the following technical result, which we prove in Subsection 
\ref{sec:HammingPlusWorstCaseBoundUPDown} below.

\begin{proposition}\label{pro:HammingPlusWorstCaseBoundUPDown}
For any distribution $\nu:\{\pm 1\}^{\Lambda}\times \{\pm 1\}^{\Lambda}$, coupling of the marginals 
$\mu^{+}_{\Lambda}(\cdot)$  and $\mu^{-}_{\Lambda}(\cdot)$, the following is true: 

\begin{align}
\lnorm\mdownup^{+}(\cdot)-\mdownup^{-}(\cdot)\rnorm_{\rm TV} &\le 
\sum_{u\in \Lambda}\Exp_{(\bsigma, \btau)\sim \nu}[ {\bf 1}\{\bsigma(u)\neq \btau(u) \}] \cdot \max_{\eta\in \{\pm 1\}^{\Lambda}}
\left\| \mu^{\Lambda, \eta}_r (\cdot) - \mu^{\Lambda, \eta_u}_r (\cdot )\right\|_{\rm TV},
\end{align}
where $\eta_u\in \{\pm 1\}^{\Lambda}$  is obtained by $\eta$ by changing the configuration
at $u$, from $\eta(u)$ to its opposite. 
\end{proposition}

Proposition \ref{prop:UpDownDistance} follows by bounding appropriately the r.h.s. of the inequality above. 
Specifically, from Theorem \ref{thrm:PairwiseInfluenceBound} we have that
\begin{align}\label{eq:FirstBound4prop:UpDownDistance}
 \max_{\eta\in \{\pm 1\}^{\Lambda}}
\left\| \mu^{\Lambda, \eta}_r (\cdot) - \mu^{\Lambda, \eta_u}_r (\cdot )\right\|_{\rm TV}
& \leq  \prod_{e\in\mathrm{path}(r,u)} \Inf_e \enspace,
\end{align}
where recall that $\mathrm{path}(r,u)$ denotes the set of edges on the path from the root $r$ to the vertex $u\in \Lambda$. 

Moreover, we show that for any $u\in \Lambda$ we have 
\begin{align}\label{eq:SecondBound4prop:UpDownDistance}
\Exp_{(\bsigma, \btau)\sim \nu}[ {\bf 1}\{\bsigma(u)\neq \btau(u) \}] & \leq  \prod_{e\in\mathrm{path}(r,u)} \Inf_e \enspace.
\end{align}
It is immediate that Proposition \ref{prop:UpDownDistance} follows from Proposition \ref{pro:HammingPlusWorstCaseBoundUPDown} 
and \eqref{eq:FirstBound4prop:UpDownDistance}, \eqref{eq:SecondBound4prop:UpDownDistance}.

We prove \eqref{eq:SecondBound4prop:UpDownDistance} by   explicitly describing a coupling $\nu$ 
that achieves the aforementioned bound. We call this coupling the ``Down Coupling". 

\subsection*{Down Coupling}
Recall that we want to couple the distributions $\mu^+_{\Lambda}$ and $\mu^-_{\Lambda}$. Instead, we couple  $\mu^+$ and $\mu^-$, i.e., rather than
coupling the conditional Gibbs marginals at $\Lambda$, we couple the conditional measure. The coupling of  $\mu^+$ and $\mu^-$, trivially, specifies a coupling
for their marginals at $\Lambda$.

Write   $\zeta:\{\pm 1\}^V\times \{\pm 1\}^V\to [0,1]$ for the coupling of $\mu^+$ and $\mu^-$ we wish to define.  We specify $\zeta$ by describing how we 
generate two configurations $(\bsigma, \btau)\in \{\pm 1\}^V\times \{\pm 1\}^V$ which are distributed as in~$\zeta$. We generate the two 
configurations inductively.   In order to specify the configurations for the vertices at level $i$ of the tree, 
we use the configurations at level $i-1$. Suppose  that we need to decide the configuration for vertex $w$, while 
we already have the configurations for vertex $v$, the parent of $w$, i.e., we have both $\bsigma(v)$ and $\btau(v)$. Then, we use  
{\em maximal coupling} for the configuration at vertex $w$, i.e.,  couple  the distributions $\mu_{w}(\cdot \ |\ \bsigma(v))$
and $\mu_{w}(\cdot \ |\ \btau(v))$ so that $\Pr[\bsigma(w)\neq \btau(w)\ |\ \bsigma(v), \btau(v)]$ is {\em minimized}. This implies that
\begin{align*}
\Pr[\bsigma(w)\neq \btau(w)\ |\ \bsigma(v), \btau(v)] &=
\lnorm \mu_{w}(\cdot\mid \bsigma(v))-\mu_{w}(\cdot\mid \btau(v))\rnorm_{\rm TV} \enspace.
\end{align*}

With the above coupling, we need to find an upper bound for $\Exp_{(\bsigma, \btau)\sim \nu}[ {\bf 1}\{\bsigma(u)\neq \btau(u) \}]$,
where $u\in \Lambda$.  Ideally,  we would like to get the one in \eqref{eq:SecondBound4prop:UpDownDistance}.

For two vertices in the tree, $v$ and $w$, such that  $v$ is the parent of $w$,  we have the following:
Given $\bsigma(v), \btau(v)$, then in the above coupling we have that 
\begin{align*}
 \Pr[\bsigma(w)\neq \btau(w)\ |\ \bsigma(v), \btau(v)] &= 
 \left\{
 \begin{array}{ccl}
 \lnorm \mu_{w}(\cdot\mid \bsigma(v) =+1)-\mu_{w}(\cdot\mid \btau(v) =-1)\rnorm_{\rm TV} &  & \textrm{ if $\bsigma(v)\neq \btau(v)$}\\
  0 &  & \textrm{ if $\bsigma(v)=\btau(v)$}
 \end{array}
 \right. , 
\end{align*}
whereas,  
\begin{align*}
 \lnorm \mu_{w}(\cdot\mid \bsigma(v) =+1)-\mu_{w}(\cdot\mid \btau(v) =-1)\rnorm_{\rm TV} &=
  \left| \mu_{w}(-1 \mid \bsigma(v) =+1)-\mu_{w}(-1 \mid \btau(v) =-1)\right|
\\
&= \left|\frac{1}{1+\exp(J_{\{w,v\}})} - \frac{\exp(J_{\{w,v\}})}{1+\exp(J_{\{w,v\}})}\right|
\\
&=\Gamma_{\{w,v\}}\enspace,
\end{align*}
where the last equality is due to Claim~\ref{cl:WhatIsGamma}.
All the above imply that if the coupling generates a disagreement at vertex $v$, i.e., $\bsigma(v)\neq \btau(v)$, 
then the disagreement propagates at $w$ with probability $\Gamma_{\{w,v\}}$.  From this point on, it is elementary to verify that
\eqref{eq:SecondBound4prop:UpDownDistance} is true, concluding the proof of Proposition \ref{prop:UpDownDistance}.

\subsection{Proof of Proposition \ref{pro:HammingPlusWorstCaseBoundUPDown}}\label{sec:HammingPlusWorstCaseBoundUPDown}
Recall that $\Lambda$ is the set of vertices at distance $h$ from the root.  Consider an enumeration of the vertices in 
$\Lambda$, e.g., we have $w_1,  \ldots, w_{\ell}$, where $\ell=|\Lambda|$. 
For any two configurations $\sigma, \tau \in \{\pm1\}^{\Lambda}$, and any $i\in \{1, \ldots, \ell\}$, 
we define the {\em interpolating  sequence}  $\mathcal{I}(\sigma, \tau)=\{\xi_i\}_{i=0,\ldots, \ell}$ as follows:
for any $i=0,\ldots, \ell$ we have $\xi_{i}\in \{\pm 1\}^{\Lambda}$ such that
\begin{align*}
\xi_i(w_j)=\sigma(w_j) &\quad \textrm{, for all } j>i ,& \textrm{and} && 
\xi_i(w_j)=\tau(w_j) &\quad \textrm{, for all } j \leq i \enspace.
\end{align*}
Note that $\xi_0=\sigma$, while $\xi_{\ell}=\tau$.  Also note that any two $\xi_i$ and $\xi_{i+1}$ may be equal. 

\begin{lemma}\label{lemma:UpDownBoundInterpol}
For any distribution $\nu:\{\pm 1\}^{\Lambda}\times \{\pm 1\}^{\Lambda}$, coupling of the marginals 
$\mu^{+}_{\Lambda}(\cdot)$  and $\mu^{-}_{\Lambda}(\cdot)$, the following is true: 

Consider $(\bsigma, \btau)$ distributed as in $\nu$, and consider also the interpolating sequence  
\begin{align}
\mathcal{I}(\bsigma,\btau)=\{\bxi_i\}_{i=0,\ldots, \ell} \enspace, 
\end{align}
that is induced by $\bsigma$, $\btau$. We have that
\begin{align}
\lnorm\mdownup^{+}(\cdot)-\mdownup^{-}(\cdot)\rnorm_{\rm TV} &\le 
 \sum^{\ell}_{i=1} \Exp \left[ {\bf 1}\{\bxi_{i-1}\neq \bxi_i\} \cdot
	\left\| \mu^{\Lambda, \bxi_{i-1}}_r (\cdot) - \mu^{\Lambda,\bxi_i}_r (\cdot )\right\|_{\rm TV}
	\right]\enspace,
\end{align}
where the expectation is with respect to  $\bxi_i$ and $\bxi_{i-1}$ of the interpolating sequence 
$\cI(\bsigma, \btau)$.
\end{lemma}

\begin{proof} From Lemma \ref{lem:NoClue}, we have that 
\begin{align*}
\lnorm\mdownup^{+}(\cdot)-\mdownup^{-}(\cdot)\rnorm_{\rm TV}&\le
\Exp_{(\bsigma, \btau) \sim \nu} \left[
	\left\|  \mu^{\Lambda, \bsigma}_{r} (\cdot) - \mu^{\Lambda, \btau}_{r} (\cdot)\right\|_{\rm TV}
	\right] \\
&=\Exp_{(\bsigma, \btau) \sim \nu} \left[
	\left\| \sum^{\ell}_{i=1} \mu^{\Lambda, \bxi_{i-1}}_{r} (\cdot) - \mu^{\Lambda, \bxi_i}_{r} (\cdot)\right\|_{\rm TV}
	\right] \\
&\leq \Exp_{(\bsigma, \btau) \sim \nu} \left[
	 \sum^{\ell}_{i=1} \left\|\mu^{\Lambda, \bxi_{i-1}}_{r} (\cdot) - \mu^{\Lambda, \bxi_i}_{r} (\cdot)\right\|_{\rm TV}
	\right] \enspace. 
\end{align*}
The last inequality above follows from triangle inequality.  Furthermore, the last inequality is equivalent to the following one:
\begin{align}\nonumber
\lnorm\mdownup^{+}(\cdot)-\mdownup^{-}(\cdot)\rnorm_{\rm TV}
&\leq \Exp_{(\bsigma, \btau) \sim \nu} \left[ 
	 \sum^{\ell}_{i=1}{\bf 1}\{\bxi_{i-1}\neq \bxi_i\} \left\|\mu^{\Lambda, \bxi_{i-1}}_{r} (\cdot) - \mu^{\Lambda, \bxi_i}_{r} (\cdot)\right\|_{\rm TV}
	\right].
\end{align}
The lemma follows by applying the linearity of expectation on the inequality above. 
\end{proof}

From Lemma \ref{lemma:UpDownBoundInterpol} we get the following:
\begin{align*}
\lnorm\mdownup^{+}(\cdot)-\mdownup^{-}(\cdot)\rnorm_{\rm TV} &\le 
 \sum^{\ell}_{i=1} \Exp \left[ {\bf 1}\{\bxi_{i-1}\neq \bxi_i\}
	\left\| \mu^{\Lambda, \bxi_{i-1}}_r (\cdot) - \mu^{\Lambda,\bxi_i}_r (\cdot )\right\|_{\rm TV}
	\right]\\
		&\leq  \sum^{\ell}_{i=1} \Exp \left[ {\bf 1}\{\bxi_{i-1}\neq \bxi_i\} 	\right] \cdot  \max_{\eta\in \{\pm 1\}^{\Lambda}}
\left\| \mu^{\Lambda, \eta}_r (\cdot) - \mu^{\Lambda, \eta_{w_i}}_r (\cdot )\right\|_{\rm TV}\\
	&=  \sum^{\ell}_{i=1} \Exp \left[ {\bf 1}\{ \bsigma(w_i)\neq \btau(w_i)\} \right] \cdot 	 \max_{\eta\in \{\pm 1\}^{\Lambda}}
\left\| \mu^{\Lambda, \eta}_r (\cdot) - \mu^{\Lambda, \eta_{w_i}}_r (\cdot )\right\|_{\rm TV}.
\end{align*}
The last derivation follows by noting that $\bxi_{i-1}\neq \bxi_{i}$, if and only if, we have 
$\bsigma(w_i)\neq \btau(w_i)$. All the above conclude the proof of Proposition \ref{pro:HammingPlusWorstCaseBoundUPDown}.

\section{Proof of Theorem~\ref{thrm:GeneralGWtree} - Proof of Non-Reconstruction for Galton-Watson}\label{sec:thrm:GeneralGWtree}

First, let us briefly recall what we want to prove.
For  any real numbers $d>0$, and $\beta>0$, for any distribution $\phi$ on   $\mathbb{R}$
 let $\Delta_{\rm KS}=\Delta_{\rm KS}(\beta,\phi)$ be defined as in~\eqref{eq:DeltaKSBPhi}.
For any offspring distribution 
 $\zeta$ (on $\mathbb{Z}_{\geq 0}$), with expectation $d$ and {\em bounded second moment}, let $\bT$ be the Galton-Watson tree with offspring distribution $\zeta$,  and let the Gibbs distribution 
$\bmu=\bmu_{\beta,\phi}$ be defined as in \eqref{def:SpinGlass} ,  on the tree $\bT$. 

We want to show that  if $d<\Delta_{\rm KS}$, then $\bmu$  exhibits non-reconstruction  i.e.,
\begin{align}\nonumber
\limsup_{h\to\infty}  \mathbb{E}_{\bT } \left[  \ 
\mathbb{E}_{\bmu } \left [  \lnorm \bmu_{h}(\cdot \ |\ \bsigma(r)=+1) -\bmu_{h}(\cdot \ |\ \bsigma(r)=-1)  \rnorm_{\rm TV} \ | \ \bT \right] 
\ \right]=0 \enspace.
\end{align}
Using Theorem \ref{thrm:ReconBoundNonUniform}, which holds for arbitrary trees and taking expectations we have that 
\begin{align}\label{eq:UpperBoundCorToGW}
\Exp_{\bT,\bmu}\left[\left(\lnorm \bmu^{+}_h(\cdot) -\bmu^{-}_h(\cdot) \rnorm_{\rm TV} \right)^2\right] &\le 
\Exp_{\bT,\bmu}\left[\sum_{v\in \Lambda}\;\prod_{e\in\mathrm{path}(r,v)} \bInf_e^2\right]\enspace,
\end{align}
where $\Lambda$ denotes the set of vertices at distance $h$ from the root. 
Working out the r.h.s. of \eqref{eq:UpperBoundCorToGW}, using the law of total expectation, while conditioning on $\bT$, we get
\begin{align}\label{eq:BeforeToBrakeGw}
\Exp_{\bT,\bmu}\left[\sum_{v\in \Lambda}\;\prod_{e\in\mathrm{path}(r,v)} \bInf_e^2\right] 
	= 
\Exp_{\bT}\left[\Exp_{\bmu}\left[\sum_{v\in \Lambda}\;\prod_{e\in\mathrm{path}(r,v)} \bInf_e^2\middle| \bT\right]\right]
	= 
\Exp_{\bT}\left[|\Lambda|\cdot \left(\Exp_{\bmu}\left[\;\bInf_e^2\right]\right)^h\right]
\enspace,
\end{align}
since for fixed $\bT$, the influences  $\bInf_e$ are independent, and identically distributed. Recalling now that $\Delta_{\rm KS} = \left( \mathbb{E}_\mu[\bInf_e^2 ] \right)^{-1}$, (notice that $\Delta_{\rm KS}$ does not depend on $\bT$), and that the offspring distributions of the vertices of $\bT$ is $\zeta$ with expectation $d$, we can rewrite  \eqref{eq:BeforeToBrakeGw} as
\begin{align*}
\Exp_{\bT,\bmu}\left[\sum_{v\in \Lambda}\;\prod_{e\in\mathrm{path}(r,v)} \bInf_e^2\right] 
	= 
\Exp_{\bT}\left[|\Lambda|\cdot \left(\Delta_{\rm KS}\right)^{-h}\right] 
	=\Delta_{\rm KS}^{-h} \cdot d^h
\enspace.
\end{align*}

Per our assumption $\Delta_{\rm KS} > d$,  there exists  $\varepsilon \in (0,1]$ such that $(1-\varepsilon)\cdot\Delta_{\rm KS} =d$. Combining the above with \eqref{eq:UpperBoundCorToGW} we get that
\begin{align*}
\Exp_{\bT,\bmu}\left[\left(\lnorm \bmu^{+}_h(\cdot) -\bmu^{-}_h(\cdot) \rnorm_{\rm TV} \right)^2\right] \le 
\Delta_{\rm KS}^{-h} \cdot d^h \le (1-\varepsilon)^h \enspace.
\end{align*}
Invoking  Markov's inequality, similarly to the proof of the non-reconstruction claim of 
Theorem \ref{thrm:GeneralDeltary} in Section \ref{sec:thrm:GeneralDeltary}, we get that 
\[
\mathbb{E}_{\bT } \left[  \ 
\mathbb{E}_{\bmu } \left [  \lnorm \bmu_{h}(\cdot \ |\ \bsigma(r)=+1) -\bmu_{h}(\cdot \ |\ \bsigma(r)=-1)  \rnorm_{\rm TV} \ | \ \bT \right] 
\ \right]
\le 2(1-\varepsilon)^{\frac{h}{4}} \enspace,
\]
so that taking limits as $h$ goes to infinity, gives the desired result.

\section{Theorem \ref{thrm:GeneralDeltary} - Proof of reconstruction.}\label{sec:ProofGeneralReconstruction}

Here we prove the reconstruction part of Theorem \ref{thrm:GeneralDeltary}. Before we delve into the proof, let us recall our setup.  For an integer $\Delta>0$, let $T=(V,E)$ be the $\Delta$-ary tree rooted at $r$. 
Let also  $\phi$ be a distribution on   $\mathbb{R}$, and let $\{\bJ_e\}_{e\in E}$ be i.i.d. random variables, each distributed as in $\phi$. For a real number $\beta \ge 0$, recall that the probability measure $\bmu=\bmu_{\beta,\phi}(\sigma)$ on $\{\pm1\}^V$ is defined by 
\begin{align}\label{def:XanaSpinGlass}
\bmu_{\beta,\phi}(\sigma) &\propto \textstyle \exp\left( \beta\sum_{\{u,w\}\in E} {\bf 1}\{\sigma(u)=\sigma(w)\} \cdot \bJ_{\{u,w\}}   \right) \enspace .
\end{align}
In this setting, for $\bJ$ distributed as 
in $\phi$ we define
\begin{align}\label{eq:XANADeltaKSBPhi}
\Delta_{\rm KS}(\beta,\phi) &= 
\textstyle \left( \mathbb{E}\left [   \left (\frac{1-\exp(\beta \bJ)}{1+\exp(\beta \bJ)}  \right)^2 \right] \right)^{-1}\enspace,
\end{align}
where the expectation is over the random variable $\bJ$.
For an integer $h > 0 $, we write $\Lambda$ for the set of vertices at distance $h$ from the root $r$. We also write 
$\bmu_{\Lambda}^{+}$, and $\bmu_{\Lambda}^{-}$ for the marginal of $\bmu_{\beta,\phi}$ on the set $\Lambda$ conditioned on root being $+1$ and $-1$, respectively. We want to prove that if  $\Delta>\Delta_{\rm KS}$, then $\bmu_{\beta,\phi}$ exhibits reconstruction, i.e., 

\begin{align}\label{eq:ReconWantToProve}
\limsup_{h\to\infty} \mathbb{E} \left [ \;\lnorm \bmu_{\Lambda}^{+}(\cdot) -\bmu_{\Lambda}^{-}(\cdot)  \rnorm_{\rm TV} \;\right]>0 \enspace.
\end{align}

To distinguish between the two layers of randomness considered here (spin configurations are random having distribution $\bmu$, and $\bmu$ itself is random as $\bJ$ is a random variable), we use $\mExp{\cdot }$ to denote expectation with 
respect to the measure $\bmu$, and reserve $\Exp[\cdot]$ for expectations taken with respect to the random variable of the couplings $\{\bJ_e\}$.

In the same spirit as in \cite{EvKenPerSchulConConIsingTree}, we show that
in order to establish \eqref{eq:ReconWantToProve}, it is sufficient to find a real function  on $\{\pm 1\}^\Lambda$ whose expected values with respect to measures $\bmu_{\Lambda}^{+}$, and $\bmu_{\Lambda}^{-}$, differ significantly, while its second moment with respect to 
$\bmu$  is not much larger than the square of the first moment.  In particular,  we show the following technical result.

\begin{theorem}\label{thm:TVtoSecond}
For  integer $h>0$,  let $T=(V,E)$ be an arbitrary tree of height $h$, rooted at vertex~$r$, and let~$\Lambda$ be the set of vertices at distance
$h$ from $r$.  
For any $\beta> 0$, for any distribution $\phi$ on   $\mathbb{R}$, let the Gibbs distribution $\bmu=\bmu_{\beta, \phi}$ on $T$ be defined as in \eqref{def:XanaSpinGlass}. 

Then, for any real function $G: \{\pm 1\}^\Lambda \mapsto \mathbb{R}$ defined on spin configurations of $\Lambda$, we have that
\begin{equation}\label{eq:TVtoSecond}
	\Exp\left[\lnorm \bmu_\Lambda^{+}(\cdot) -\bmu_\Lambda^{-}(\cdot)  \rnorm_{\rm TV} \right]
	\ge
	\frac{\left(\Exp\left[\mExp{G}_{\bmu_\Lambda^+} - \mExp{G}_{\bmu_\Lambda^-}\right]\right)^2}
		{4\Exp\left[\mExp{G^2}_{\bmu}\right]} \enspace,
\end{equation}
{where $\Exp$ is with respect to the couplings $\{\bJ_e\}$ on the edges of $T$, induced the measure $\mathbold{\mu}$.}
\end{theorem}

The proof of Theorem~\ref{thm:TVtoSecond} appears in Section~\ref{sec:ProofOfSecontToTV}. We now wish to define a real function $F_h$ on spin configurations of $\Lambda$, whose ratio in the r.h.s. of~\eqref{eq:TVtoSecond} is bounded away from zero. To this end, we define the ``signed influence" of an edge $e \in E$ to be
\begin{equation}\label{eq:defSignedGamma}
\textstyle
\widehat{\bInf_e} = 
	\frac {\exp(\beta \bJ_e) - 1}
		{1+ \exp(\beta \bJ_e)}\enspace.
\end{equation}
Observe that due to Claim~\ref{cl:WhatIsGamma}, we have that the relationship between $\widehat{\bInf_e}$, defined in above, and $\bInf_e$, defined by~\eqref{eq:InfDef} in Section~\ref{sec:TreeRecurInf}, is simply 
\begin{equation}\label{eq:SignedVsNormalGamma}
\textstyle
\left|\widehat{\bInf_e}\right| = {\bInf_e}
\enspace.
\end{equation}

\begin{definition}\label{def:DefOfFlipMaj} For integer $h> 0 $, let $T=(V,E)$ be an arbitrary tree, rooted at vertex~$r$, and let  $\Lambda$ be the set of vertices of $T$ at distance $h$ from the root. 
For any $\beta\ge 0$, for any distribution $\phi$ on 
 $\mathbb{R}$, let the Gibbs distribution $\bmu=\bmu_{\beta, \phi}$ on $T$ be defined as in \eqref{def:XanaSpinGlass}.

The \emph{flipped majority vote}  is the function $F_h : \{\pm1\}^\Lambda \to \mathbb{R}$ with
\begin{equation}\label{eq:DefOfFlipMaj}
F_h(\tau) = \sum_{u \in \Lambda} \tau(u)  \times \prod_{e \in \mathrm{path}(r,u)} \widehat{\bInf_e} \enspace,
\end{equation}
where $\widehat{\bInf_e}$ be defined as in \eqref{eq:defSignedGamma}.
\end{definition}

The following proposition expresses the enumerator and denominator of ratio in the r.h.s. of~\eqref{eq:TVtoSecond} for the flipped majority vote, $F_h$, defined above, in terms of the edge influences $\bInf_e$.  For two vertices $u,v$ of $T$, write $\mathrm{path}(u,v)$ for the set of edges along the 
unique path between $u$ and $v$, and write $u\wedge v$ for the common ancestor of $u$ and $v$ farthest from the root $r$.

\begin{proposition}\label{prop:RatioOfFh} For  integer $ h>0$,  let $T=(V,E)$ be an arbitrary tree of height $h$, rooted at vertex~$r$, and let  $\Lambda$ be the set of vertices at distance $h$ from the root of  $T$.  
For any $\beta > 0$, for any distribution $\phi$ on $\mathbb{R}$, let the Gibbs distribution $\bmu=\bmu_{\beta, \phi}$ on $T$ be defined as in \eqref{def:XanaSpinGlass}. 

Then,
\begin{equation}\label{eq:EqOfRarioOfFhNum}
	\Exp\left[\mExp{F_h}_{\bmu_\Lambda^+} - \mExp{F_h}_{\bmu_\Lambda^-}\right]
	=
	2{\sum_{v \in \Lambda} \;\prod_{e \in \mathrm{path}(r,v)}\Exp\left[ \bInf^2_e\right]}\enspace,
\end{equation}
and
\begin{equation}\label{eq:EqOfRarioOfFhDen}
		{\Exp\left[\mExp{F_h^2}_{\bmu}\right]}
	=
		{\displaystyle\sum_{u,v \in \Lambda}\;\left(\prod_{e \in \mathrm{path}(u,v)}\Exp\left[ \bInf_e^2\right]\right)\left(\prod_{e^{\prime}\in\mathrm{path}(r,u\wedge v) }\Exp\left[\bInf_{e^{\prime}}^2\right]\right)}\enspace,
\end{equation}
where $\bInf_e$ is the influence of edge $e$ defined in \eqref{eq:whatisgamma}.
\end{proposition}

The proof of Proposition~\ref{prop:RatioOfFh} appears in Section~\ref{sec:LemRatioOfFh}. Finally, for the case of $\Delta$-ary tree,  and using Proposition \ref{prop:RatioOfFh}, we have the following lemma

\begin{lemma}\label{lem:lowerBoundDtree}
For  integers $\Delta, h>0$,  let $T=(V,E)$ be the $\Delta$-ary tree of height $h$, rooted at vertex~$r$, and let  $\Lambda$ be the 
set of vertices at distance $h$ from the root of $T$.  
For any $\beta > 0$, for any distribution $\phi$ on $\mathbb{R}$, let the Gibbs distribution $\bmu=\bmu_{\beta, \phi}$ on $T$ be defined as in \eqref{def:XanaSpinGlass}. 

Let also $\Delta_{\rm KS}(\beta,\phi) = \Delta_{\rm KS}$ be defined as in~\eqref{eq:XANADeltaKSBPhi}, and let $F_h$ be the flipped majority vote defined in~\eqref{eq:DefOfFlipMaj}. Suppose that  $\Delta >\Delta_{\rm KS}$. Then, we have that
\begin{equation}\label{eq:EqOfRarioOfFhDreg}
	\frac{\left(\Exp\left[\mExp{F_h}_{\bmu_\Lambda^+} - \mExp{F_h}_{\bmu_\Lambda^-}\right]\right)^2}
		{4\Exp\left[\mExp{F_h^2}_{\bmu}\right]}
	\ge 
	\frac{\delta}{1+\delta} \enspace,
\end{equation}
where  $\delta>0$ is defined by $1+\delta=\Delta/\Delta_{\rm KS}$.
\end{lemma}

We prove Lemma~\ref{lem:lowerBoundDtree} in Section~\ref{sec:ProofLemmaBound}. The reconstruction claim of Theorem \ref{thrm:GeneralDeltary}, follows now readily from Theorem~\ref{thm:TVtoSecond} and Lemma \ref{lem:lowerBoundDtree}.

\begin{proof}[Proof of Theorem \ref{thrm:GeneralDeltary} - Reconstruction.]

Consider the Gibbs distribution $\bmu=\bmu_{\beta, \phi}$ defined as in \eqref{def:XanaSpinGlass}, on the $\Delta$-ary tree $T=(V,E)$.  
We need to show that for $\Delta>\Delta_{\rm KS}$, where $\Delta_{\rm KS} = \Delta_{\rm KS}(\beta,\phi)$ is defined as in~\eqref{eq:XANADeltaKSBPhi}, we have that 
\begin{equation}\label{eq:Target4ReconThrm:GeneralDeltary}
\limsup_{h \to \infty}
\Exp\left[  \lnorm \bmu_\Lambda^+ -\bmu_\Lambda^-  \rnorm_{\rm TV} \right] >0 \enspace,
\end{equation}
where $\Exp$ is taken with respect to the random variables $\{\bJ_e\}$. Let now $F_h:\{\pm1\}^\Lambda \to \mathbb{R}$ be the flipped majority vote defined as in \eqref{eq:DefOfFlipMaj}. Applying Theorem~\ref{thm:TVtoSecond}, which holds for any real function on $\{\pm1\}^\Lambda $, on $F_h$, gives

\begin{equation*}\label{eq:TVtoSecondForFlipped}
	\Exp\left[\lnorm \bmu_\Lambda^{+}(\cdot) -\bmu_\Lambda^{-}(\cdot)  \rnorm_{\rm TV} \right]
	\ge
	\frac{\left(\Exp\left[\mExp{F_h}_{\bmu_\Lambda^+} - \mExp{F_h}_{\bmu_\Lambda^-}\right]\right)^2}
		{4\Exp\left[\mExp{F_h^2}_{\bmu}\right]} \enspace.
\end{equation*}
Since $\Delta>\Delta_{\rm KS}$, there exist a $\delta > 0$ such that $\Delta=(1+\delta)\Delta_{\rm KS}$. Applying Lemma \ref{lem:lowerBoundDtree} on the r.h.s. of the above gives further that
\begin{equation*}\label{eq:TVBoundForFlipped}
	\Exp \left[\lnorm \bmu_\Lambda^{+}(\cdot) -\bmu_\Lambda^{-}(\cdot)  \rnorm_{\rm TV} \right]
	\ge
	\frac{\delta}{1+\delta}\enspace.
\end{equation*}
Taking limits, yields trivially 
\begin{equation}\nonumber 
	\limsup_{h \to \infty}
	\Exp\left[
	\lnorm \bmu_\Lambda^{+}(\cdot) -\bmu_\Lambda^{-}(\cdot)  \rnorm_{\rm TV} 
	\right]
	\ge
	\frac{\delta}{1+\delta}>0\enspace,
\end{equation}
as desired, concluding the proof of  the reconstruction claim of Theorem~\ref{thrm:GeneralDeltary}.
\end{proof}

\section{Proof Of  Theorem~\ref{thm:TVtoSecond}}\label{sec:ProofOfSecontToTV}

Let us first introduce some additional notation. For  integer  $h>0$,  let $T=(V,E)$ be an arbitrary tree of height $h$ rooted at vertex~$r$, and let  $\Lambda$ be the set of vertices at distance $h$ from the root  of $T$. Given a function  $G: \{\pm 1\}^\Lambda \mapsto \mathbb{R}$ defined on the spin configurations of $\Lambda$, write $\range(G) \subseteq  \mathbb{R}$ for the range of $G$. For $t \in\range(G)$, and $s\in \{\pm 1\}$, let us define
\begin{align}
	{\bmu_{r}(s\mid \{G=t\})=\bmu(\bsigma(r) = s \mid \{G=t\})}\;, 
		&&\text{and}&& 
	{\bmu^{s}( \{G=t\})=\bmu( \{G=t\}\mid \bsigma(r) = s)} \enspace,
\end{align}
where we used the notation 
$
\{G=t\} = \left\{\tau \in \{\pm 1\}^\Lambda: G(\tau) = t\right\}
$. 

Recall also that $\bmu_{r}$ denotes the marginal of $\bmu$ at the root $r$, while $\bmu_\Lambda^{+}$ and 
	$\bmu_\Lambda^{-}$ denote the marginals of $\bmu$ at $\Lambda$  conditioned on $\bsigma(r) =+1$ and 
	$\bsigma(r) =-1$, respectively.

\begin{lemma}\label{lem:LowerTV}
For  integer  $h>0$,  let $T=(V,E)$ be arbitrary tree of height $h$ rooted at vertex~$r$, and let  $\Lambda$ be 
the set of vertices at distance $h$ from the root of   $T$. 
For any $\beta\ge 0$, for any distribution $\phi$ on $\mathbb{R}$, let the Gibbs distribution $\bmu=\bmu_{\beta, \phi}$ on $T$ be defined as in \eqref{def:XanaSpinGlass}. 

Then, for any $G: \{\pm 1\}^\Lambda \mapsto \mathbb{R}$, we have that
\begin{equation}\label{eq:LowerTV}
	\mExp{\left|\bmu_{r}(+1\mid \{G=t\}) -\bmu_{r}(-1\mid \{G=t\})  \right|}_{t\sim \bmu} 
		\le
	\lnorm \bmu_\Lambda^{+}(\cdot) -\bmu_\Lambda^{-}(\cdot)  \rnorm_{\rm TV}
	\enspace.
	\end{equation}
\end{lemma}

\begin{proof}
First, we observe that the l.h.s. of \eqref{eq:LowerTV} can be expressed as total variation distance. In particular, we have that
\begin{equation}\label{eq:ItsATV}
	\mExp{\left|\bmu_{r}(+1\mid \{G=t\}) -\bmu_{r}(-1\mid \{G=t\})  \right|}_{t\sim \bmu} 
		=
	\lnorm \bmu^{+}( \{G=\cdot\}) -\bmu^{-}( \{G=\cdot\})  \rnorm_{\rm TV}
	\enspace,
\end{equation}
where recall that  $\bmu_{r}$ is the marginal of $\bmu$ at the root $r$, while $\bmu^{+}$ and $\bmu^{-}$ denote the measure 
$\bmu$ conditional on $\bsigma(r) =+1$ and $\bsigma(r) =-1$, respectively. Indeed, we have

\begin{align}
\nonumber
&\mExp{\left|\bmu_{r}(+1\mid \{G=t\}) -\bmu_{r}(-1\mid \{G=t\})  \right|}_{t\sim \bmu} \\ \nonumber
	&=
	\sum_{t\in \range(G)} 
		\left|\bmu_{r}(+1\mid \{G=t\}) - \bmu_{r}(-1\mid \{G=t\})\right|\cdot{\bmu(\{G=t\})} \\ \nonumber
	&= 		
	\sum_{t\in \range(G)} 
		\left|\bmu_{r}(+1\mid \{G=t\}) {\bmu(\{G=t\})} 
				- \bmu_{r}(-1\mid \{G=t\}){\bmu(\{G=t\})} \right|\\\label{eq:BayesRuleEqA}
	&= 
	\frac{1}{2}  \sum_{t\in \range(G)} 
		\left|\bmu_{r}(+1\mid \{G=t\}) \frac{\bmu(\{G=t\})}{\bmu(\bsigma(r) =+1)} 
				- \bmu_{r}(-1\mid \{G=t\})\frac{\bmu(\{G=t\})}{\bmu(\bsigma(r) =-1)} \right|\\ \label{eq:BayesRuleEqB}
		&= 
		\frac{1}{2}  \sum_{t\in \range(G)} 
			\left|\bmu^{+}( \{G=t\}) - \bmu^{-}( \{G=t\})\right|\\ \nonumber
		&= 
		\lnorm \bmu^{+}(\{G=\cdot\}) -\bmu^{-}( \{G=\cdot\})  \rnorm_{\rm TV} \enspace,
\end{align}
where \eqref{eq:BayesRuleEqA} and \eqref{eq:BayesRuleEqB} follow from Bayes' rule, and the fact that 
${\bmu(\bsigma(r) =-1)= \bmu(\bsigma(r) =+1)={1}/{2}}$.

Recall now that the total variation distance of two measures $p$ and $q$, defined on the same probability space $(\Omega, \mathcal{F})$, can be equivalently defined as
\begin{equation*}\label{eq:AnotherDefTV}
\lnorm p-q \rnorm_{\rm TV} = \sup_{ A\in \mathcal{F}} |p(A)-q(A)| \enspace.
\end{equation*}
Given a subalgebra $\mathcal{G} \subseteq \mathcal{F}$, let $p^\prime$ and $q^\prime$ denote the restrictions of $p$ and $q$ on $\mathcal{G}$, respectively. Then 

\begin{equation*}\label{eq:subAlgTV}
		\lnorm p^\prime-q^\prime \rnorm_{\rm TV} 
			= 
		\sup_{ A\in \mathcal{G}} |p^\prime(A)-q^\prime(A)| \le \sup_{ A\in \mathcal{F}} |p(A)-q(A)| 
			= 
		\lnorm p-q \rnorm_{\rm TV} \enspace.
\end{equation*}
Observe now that $\bmu^{+}( \{G=\cdot\})$ and $\bmu^{-}( \{G=\cdot\})$ are precisely the restrictions of $ \bmu_\Lambda^{+}(\cdot)$ and  $\bmu_\Lambda^{-}(\cdot)$ on the $\sigma$-algebra generated by the function $G$, respectively. Hence, per the above and \eqref{eq:ItsATV}, we have that
\begin{multline*}
	\mExp{\left|\bmu_{r}(+1\mid \{G=t\})-\bmu_{r}(-1\mid \{G=t\})  \right|}_{t\sim \bmu} 
		\\
	=
		\lnorm \bmu^{+}( \{G=\cdot\}) -\bmu^{-}( \{G=\cdot\})  \rnorm_{\rm TV} 	\le
		\lnorm \bmu_\Lambda^{+}(\cdot) -\bmu_\Lambda^{-}(\cdot)  \rnorm_{\rm TV} \enspace.
\end{multline*}
The above concludes the proof of Lemma \ref{lem:LowerTV}. 
\end{proof}

As usual, it is easier to handle squares than absolute values. Observing that
\begin{equation*}
0 \le \mExp{\left|\bmu_{r}(+1\mid \{G=t\}) -\bmu_{r}(-1\mid \{G=t\})  \right|}_{t\sim \bmu} \le 1 \enspace,
\end{equation*}
we have that
\begin{equation*}
\mExp{\left(\bmu_{r}(+1\mid \{G=t\}) -\bmu_{r}(-1\mid \{G=t\})  \right)^2}_{t\sim \bmu} \le \mExp{\left|\bmu_{r}(+1\mid \{G=t\}) -\bmu_{r}(-1\mid \{G=t\})  \right|}_{t\sim \bmu}
\enspace,
\end{equation*}
which further implies
\begin{equation}\label{eq:StepA4TVtoSecond}
	\mExp{\left(\bmu_{r}(+1\mid \{G=t\}) -\bmu_{r}(-1\mid \{G=t\})  \right)^2}_{t\sim \bmu} 
		\le
	\lnorm \bmu_\Lambda^{+}(\cdot) -\bmu_\Lambda^{-}(\cdot)  \rnorm_{\rm TV}
	\enspace.
\end{equation}

Finally, we need to prove the following lemma.

\begin{lemma}\label{lem:YuvalSecSpin}
For  integer  $h>0$,  let $T=(V,E)$ be arbitrary tree of height $h$ rooted at vertex~$r$, and let  $\Lambda$ be the 
set of vertices at distance $h$ from the root of $T$. 
For any $\beta > 0$, for any distribution $\phi$ on $\mathbb{R}$, let the Gibbs distribution $\bmu=\bmu_{\beta, \phi}$ on $T$ be defined as in \eqref{def:XanaSpinGlass}. 

Then, for any $G: \{\pm 1\}^\Lambda \mapsto \mathbb{R}$, we have that
\begin{align}\label{eq:CSbound}
\Exp\left[\mExp{\left(\bmu_{r}(+1\mid \{G=t\}) -\bmu_{r}(-1\mid \{G=t\})  \right)^2}_{t\sim \mu} \right] 
	\ge 
	\frac{\left(\Exp\left[\mExp{G}_{\bmu_\Lambda^+}
		-\mExp{G}_{\bmu_\Lambda^-}\right]\right)^2}
		{4\Exp\left[\mExp{G^2}_{\bmu}\right]} 
		\enspace.
\end{align}
Recall that the expectation is taken w.r.t. the coupling parameters in $\bmu$.
\end{lemma}
\begin{proof} Expanding the enumerator of the fraction in the right hand side of \eqref{eq:CSbound} gives
\begin{align}
\nonumber
	\Exp^2&\left[\mExp{G}_{\bmu_\Lambda^+} -\mExp{G}_{\bmu_\Lambda^-}\right] \\ \nonumber
		&=
		\Exp^2
	\left[\sum_{t \in \range(G)} t\cdot\bmu^{+}( \{G=t\})- 
	\sum_{t \in \range(G)} t\cdot\bmu^{-}( \{G=t\})\right] \\ \nonumber
		&= 
		\left(\int\sum_{t \in \range(G)} t \cdot\left[\bmu^{+}( \{G=t\})
		- \bmu^{-}( \{G=t\})\right]d{\phi}\right)^2 \\ \label{eq:CSbouBayesRule}
		&= 
		\left(2\int\sum_{t \in \range(G)} t \cdot\left[\bmu_{r}(+1\mid \{G=t\})
		- \bmu_{r}(-1\mid \{G=t\})\right]\cdot \bmu\left(\{G=t\}\right) d{\phi}\right)^2 
		\enspace,
\end{align}
where the last equality follows from Bayes' rule. Applying now the Cauchy-Schwartz inequality with factors
\begin{align*}
\left[\bmu_{r}(+1\mid \{G=t\})
		- \bmu_{r}(-1\mid \{G=t\})\right]\cdot \sqrt{\bmu\left(\{G=t\}\right)}\;,
		&&\text{and}&&
	t\cdot \sqrt{\bmu\left(\{G=t\}\right)}\enspace,
\end{align*}
we further get that %
\begin{align*}
\lefteqn{
\Exp^2\left[\mExp{G}_{\mu_\Lambda^+} -\mExp{G}_{\mu_\Lambda^-}\right]
} \vspace{2cm} \\
		&\le 
		4
		\left(\int \sum_{t \in  \range(G)} t^2\bmu\left(\{G=t\}\right)d{\phi}\right)
		\cdot \int \sum_{t \in  \range(G)}\left[\bmu_{r}(+1\mid \{G=t\}) -
		 \bmu_{r}(-1\mid \{G=t\})\right]^2\bmu\left(\{G=t\}\right)d{\phi}\\
		&=
		4\cdot
		\Exp\left[\mExp{G^2}_\bmu\right] \cdot
		\Exp\left[\mExp{\left(\bmu_{r}(+1\mid \{G=t\}) -\bmu_{r}(-1\mid \{G=t\})  \right)^2}_{t\sim \bmu} \right] \enspace.
\end{align*}
The above concludes the proof of Lemma \ref{lem:YuvalSecSpin}. 
\end{proof}

 Theorem ~\ref{thm:TVtoSecond} now follows from \eqref{eq:StepA4TVtoSecond} 
 and Lemma \ref{lem:YuvalSecSpin}.

\section{Proof of Proposition~\ref{prop:RatioOfFh}}\label{sec:LemRatioOfFh}

Let $T=(V,E)$ be an arbitrary tree rooted at $r$. Let also $\phi$ be a distribution on $\mathbb{R}$, and let $\{\bJ_e\}_{e\in E}$ be i.i.d. random variables, each $\bJ_e$ distributed as in $\phi$. For a real number $\beta \ge 0$, recall that the probability measure $\bmu=\bmu_{\beta,\phi}(\sigma)$ on $\{\pm1\}^V$ is defined by 
\begin{align}\label{def:XanaKaiXanaSpinGlass}
\bmu_{\beta,\phi}(\sigma) &\propto \textstyle \exp\left( \beta\sum_{\{w,u\}\in E} {\bf 1}\{\sigma(u)=\sigma(w)\} \cdot \bJ_{\{u,w\}}   \right) \enspace .
\end{align}
Recall also that for each $e \in  E$ we have defined 
\begin{equation}\label{eq:XanAdefSignedGamma}
\textstyle
\widehat{\bInf_e} = 
	\frac {\exp(\beta \bJ_e) - 1}
		{1+ \exp(\beta \bJ_e)}\enspace,
\end{equation}
and we use $\mathrm{path}(u,v)$ to denote the set of edges along the unique path between $u$ and $v$. Finally, recall that  $\bmu^{u,s}$ denotes the measure $\bmu$ conditional on $\bsigma(u) = s$, for $s \in \{\pm 1\}$. We now have the following lemma.

\begin{lemma}\label{lem:ExpSigma}
Let $T=(V,E)$ be an arbitrary tree, and let $\bmu$ be the Gibbs measure on $T$ defined as in \eqref{def:XanaKaiXanaSpinGlass}, and
$\widehat{\bInf_e}$ be as in \eqref{eq:XanAdefSignedGamma}. Then, for any two vertices $u, w$ of $T$, and $s \in \{\pm 1\}$ we have

\begin{equation}\label{eq:ExpSigmaEq}
\mExp{\bsigma(w)}_{\bmu^{u,s}} = s \prod_{e \in \mathrm{path}(u,w)} \widehat{\bInf_e} \enspace.
\end{equation}
\end{lemma}

\begin{proof}
Let $(u, q_1, \ldots, q_t,w)$ be the unique path from $u$ to $w$, and write 
$P=\{q_1, \ldots, q_t,w\}$ for the set of vertices along that path, apart from $u$. We now see that for any $s\in \{\pm 1\}$
we have that
\begin{equation}\label{eq:FinalePathMarg}
	\mExp{\bsigma(w)}_{\bmu^{u,s}} 
		 =  \sum_{\tau \in \{\pm1\}^{P}} \tau(v) \cdot \bmu(\bsigma(P)=\tau \mid \bsigma(u) = s) \enspace.
\end{equation}

We now prove \eqref{eq:ExpSigmaEq}  by induction on the distance between $u$ and $v$. For the base case, corresponds to
 $w$ and $u$ being adjacent vertices.  Then,  for any $s\in \{\pm 1\}$, equation \eqref{eq:FinalePathMarg} becomes
\begin{align*}
\mExp{\bsigma(w)}_{\bmu^{u,s}} 
	&= \bmu(\bsigma(w)=+1\mid \bsigma(u)=s) - \bmu(\bsigma(w) = -1 \mid \bsigma(u)=s) \\
	&=\begin{cases}
	\frac{\exp(\beta \bJ_{\{u,w\}}) - 1}{ 1+ \exp(\beta \bJ_{\{u,w\}})} = \widehat{\bInf}_{\{u,w\}} &\text{ if } s = +1\\[6pt]
	\frac{1-\exp(\beta \bJ_{\{u,w\}})}{ 1+ \exp(\beta \bJ_{\{u,w\}})} = -\widehat{\bInf}_{\{u,w\}}&\text{ if } s = -1\/
	\end{cases} \enspace,
\end{align*}
as desired. 

Assume now \eqref{eq:ExpSigmaEq} holds for any pair of  vertices whose distance is at most $t$. Let $u$, $w$ be a pair of  vertices 
of distance $t+1$. In particular, let $(u, q_1,\dots, q_t, w)$ be the (unique) path from $u$ to $w$, and write $P=\{q_1, \ldots, q_t,w\}$, 
and $P^\prime =P\setminus\{q_1\}$. From \eqref{eq:FinalePathMarg} we have that for $s \in \{\pm1\}$
\begin{align*}
\mExp{\bsigma(w)}_{\bmu^{u,s}}  
	&=
	\sum_{\tau \in \{\pm1\}^{P}} \tau(w) \cdot \bmu(\bsigma(P)=\tau \mid \bsigma(u) = s)  \\ 
	&=
	\sum_{\xi\in \{\pm1\}}\sum_{\tau \in \{\pm1\}^{P^\prime}} \tau(w) \cdot \bmu\left(\bsigma(q_1)=\xi\mid \bsigma(u)=s\right)	
		\cdot \bmu\left(\bsigma(P')=\tau\mid \bsigma(q_1)=\xi,\  \bsigma(u)=s\right) \\ 
	&=
	\sum_{\xi\in \{\pm1\}}\sum_{\tau \in \{\pm1\}^{P^\prime}} \tau(w) \cdot \bmu\left(\bsigma(q_1)=\xi\mid \bsigma(u)=s\right)	
		\cdot \bmu\left(\bsigma(P')=\tau\mid \bsigma(q_1)=\xi\right) \enspace,
\end{align*}
where 
the last equality follows from the Markov property of the model. 
Pushing now forward the sum over the configurations of $P^\prime$ we further get
\begin{align*}
\mExp{\bsigma(w)}_{\bmu^{u,s}}  
	&=
	\sum_{\xi\in \{\pm1\}}
	\bmu\left(\bsigma(q_1)=\xi\mid \bsigma(u)=s\right) \cdot
	\left(
	\sum_{\tau \in \{\pm1\}^{P^\prime}} \tau(w) 	
		\cdot \bmu\left(\bsigma(P')=\tau\mid \bsigma(q_1)=\xi\right)
	\right)\\ 
	&=
	\sum_{\xi\in \{\pm s\}}
	\bmu\left(\bsigma(q_1)=\xi\mid \bsigma(u)=s\right) \cdot
	\mExp{\bsigma(w)}_{\bmu^{q_1,\xi}}  \enspace,
\end{align*}
where $\bmu^{q_1,\xi}$ denotes the measure $\bmu$ conditional on $\bsigma(q_1) = \xi$, and we get the last equality  from \eqref{eq:FinalePathMarg}. Expanding now the sum over $\xi \in \{\pm s\}$, we further get
\begin{align}
\nonumber
\mExp{\bsigma(w)}_{\bmu^{u,s}}  
	&=
	\frac{\exp(\beta J_{\{u,w\}}) }{ 1+ \exp(\beta J_{\{u,w\}})} \cdot
	\mExp{\bsigma(w)}_{\bmu^{q_1,s}}  + 	\frac{1}{ 1+ \exp(\beta J_{\{u,w\}})} \cdot
	\mExp{\bsigma(w)}_{\bmu^{q_1,(-s)}}\\ \label{eq:IndHyp}
	&=
	\left(
	{\textstyle\frac{\exp(\beta J_{\{u,w\}}) }{ 1+ \exp(\beta J_{\{u,w\}})} }\cdot
	s \cdot\prod_{e \in \mathrm{path}(q_1,w)} \widehat{\bInf_e}
	\right)
	+ 	
	\left(
	{\textstyle\frac{1}{ 1+ \exp(\beta J_{\{u,w\}})}} \cdot
	(-s) \cdot\prod_{e \in \mathrm{path}(q_1,w)} \widehat{\bInf_e} 
	\right)\\ \nonumber
	&=
	 \left(s\cdot\prod_{e \in \mathrm{path}(q_1,w)} \widehat{\bInf_e}\right)\cdot
	\left(
	\frac{\exp(\beta J_{\{u,w\}})-1}{ 1+ \exp(\beta J_{\{u,w\}})}
	 \right)\\ \label{eq:DefOfGam}
	 	&=
		s\cdot\prod_{e \in \mathrm{path}(u,w)} \widehat{\bInf_e}\enspace, 
\end{align}
where \eqref{eq:IndHyp} follows from the inductive hypothesis applied on vertices $q_1$ and $w$, while \eqref{eq:DefOfGam} follows from the definition of  $\widehat{\bInf_e}$ in \eqref{eq:XanAdefSignedGamma}.
\end{proof}

Using Lemma~\ref{lem:ExpSigma}, we now prove the following lemma about pairwise spin correlations.

\begin{lemma}\label{lem:ExpSigmaSigma}
Let $T=(V,E)$ be any finite tree, and let $\bmu$ be the Gibbs measure on $T$ defined as in \eqref{def:XanaKaiXanaSpinGlass}, and
$\widehat{\bInf_e}$ be as in \eqref{eq:XanAdefSignedGamma}. Then, for any two vertices $u, v$ of $T$, we have
\begin{equation}\nonumber 
\mExp{\bsigma(u)\cdot\bsigma(v)}_{\bmu} = \prod_{e \in \mathrm{path}(u,v)} \widehat{\bInf_e} \enspace.
\end{equation}
\end{lemma}

\begin{proof}
Indeed,
\begin{align}\nonumber
\mExp{\bsigma(u)\cdot\bsigma(v)}_{\bmu} 
	&= 
	\sum_{\tau \in \{\pm\}^{V}} \tau(u) \cdot \tau(v) \cdot\bmu(\tau)\\ \label{eq:ConditionOnU}
	&= 
	\frac{1}{2}\sum_{\tau \in \{\pm\}^{V\setminus\{u\}}} \tau(v) \cdot \bmu^{u,+}(\tau)
	- \frac{1}{2}\sum_{\tau \in \{\pm\}^{V\setminus\{u\}}} \tau(v) \cdot \bmu^{u,-}(\tau)
	\\ \nonumber
	&= 
	\frac{1}{2}\left(\mExp{\bsigma(v)}_{\bmu^{u,+}} \right) 
	- \frac{1}{2}\left(\mExp{\bsigma(v)}_{\bmu^{u,-}}\right) 
	\\ \label{eq:ByTheLemmaAb}
&=	\frac{1}{2}\left(\prod_{e \in \mathrm{path}(u,v)} \widehat{\bInf_e} \right) 
	- \frac{1}{2}\left((-1)\prod_{e \in \mathrm{path}(u,v)} \widehat{\bInf_e}\right)  
= \prod_{e \in \mathrm{path}(u,v)} \widehat{\bInf_e} \enspace, 
\end{align}
where $\bmu^{u,+}$, $\bmu^{u,-}$ denote measure $\bmu$ conditional on $u$ being $+1$ and $-1$, respectively. 
We get \eqref{eq:ConditionOnU}  by the law of total probability, i.e., we condition on the spin of $u$, and use the 
fact that
$
\bmu(\bsigma(r) =-)= \bmu(\bsigma(r) =+)={1}/{2}
$.
Also,  \eqref{eq:ByTheLemmaAb} follows from Lemma~\ref{lem:ExpSigma}.
All the above conclude the proof of Lemma \ref{lem:ExpSigmaSigma}.
\end{proof}

Recall that $\bmu_\Lambda^+$, and $\bmu_\Lambda^-$ denote the marginals of $\bmu$ on $\Lambda$, conditioned on $\bsigma(r) = +1$, and $\bsigma(r) = -1$, respectively. Finally, let $\widehat{\bInf_e}$ be the signed influence of edge $e$, defined as in \eqref{eq:XanAdefSignedGamma}, and $F_h$ be the flipped majority vote introduced in Definition \ref{def:DefOfFlipMaj}.

We start by applying Lemmas~\ref{lem:ExpSigma}, and  \ref{lem:ExpSigmaSigma}, to calculate the first moments of $F_h$ with respect to the measures $\bmu_\Lambda^+$, and $\bmu_\Lambda^-$. We have that 
\begin{align}\label{eq:FirstMomOfFh+}
	\mExp{F_h}_{\bmu^{+}_{\Lambda}} 
		= \sum_{v \in \Lambda}\mExp{\bsigma(v)}_{\bmu^{+}_\Lambda} \;\cdot
		\prod_{e \in \mathrm{path}(r,v)} \widehat{\bInf_e} 
		=\sum_{v \in \Lambda} \;\prod_{e \in \mathrm{path}(r,v)} \left(\widehat{\bInf_e}\right)^2
		=\sum_{v \in \Lambda} \;\prod_{e \in \mathrm{path}(r,v)} \bInf_e^2
		\enspace,
\end{align}
where the first equality follows from linearity of expectation, the second equality by applying Lemma~\ref{lem:ExpSigma}, and the last equality is due to \eqref{eq:SignedVsNormalGamma}. Similarly,
\begin{align}\label{eq:FirstMomOfFh-}
	\mExp{F_h}_{\bmu^{-}_{\Lambda}} 
		= \sum_{v \in \Lambda}\mExp{\bsigma(v)}_{\bmu^{-}_\Lambda} \;\cdot
		\prod_{e \in \mathrm{path}(r,v)} \widehat{\bInf_e} 
		=-\sum_{v \in \Lambda} \;\prod_{e \in \mathrm{path}(r,v)} \left(\widehat{\bInf_e}\right)^2
		=-\sum_{v \in \Lambda} \;\prod_{e \in \mathrm{path}(r,v)} \bInf_e^2
		\enspace.
\end{align}

It is now easy to derive \eqref{eq:EqOfRarioOfFhNum} as
\begin{align*}
{\Exp\left[\mExp{F_h}_{\bmu_\Lambda^+} - \mExp{F_h}_{\bmu_\Lambda^-}\right]}
=
{\Exp\left[2\sum_{v \in \Lambda} \;\prod_{e \in \mathrm{path}(r,v)} \bInf_e^2\right]}
=2 \sum_{v \in \Lambda} \;\prod_{e \in \mathrm{path}(r,v)} \Exp\left[\bInf_e^2\right] \enspace,
\end{align*}
where the first equality follows from \eqref{eq:FirstMomOfFh+} and \eqref{eq:FirstMomOfFh-}. To get the last equality, we use the linearity of expectation, and the fact that the couplings $\{\bJ_e\}_{e\in E}$, (and thus, also $\{\bInf_e\}_{e\in E}$), are independent.

We now use Lemma  \ref{lem:ExpSigmaSigma} to calculate the second moment of $F_h$ with respect to 
$\bmu_\Lambda$. Expanding $F_h^2$, we have that
\begin{align*}
	\mExp{F_h^2}_{\bmu}
	&= \sum_{u \in \Lambda} \mExp{\bsigma^2(u)}_{\bmu} \cdot \prod_{e \in \mathrm{path}(r,u)} 
	\left(\widehat{\bInf_e}\right)^2 +\sum_{\substack{u,v \in \Lambda\\ u \neq v}} \mExp{\bsigma(u)\bsigma(v)}_{\bmu} \cdot 
	\prod_{e \in \mathrm{path}(r,u)} 
	\widehat{\bInf_e}
	\cdot 
	\prod_{{e^\prime} \in \mathrm{path}(r,v)} 
	\widehat{\bInf_{e^\prime}}
	\\ 
	&= \sum_{u \in \Lambda} \;\prod_{e \in \mathrm{path}(r,u)} 
	\bInf^2_e
	+
	\sum_{\substack{u,v \in \Lambda\\ u \neq v}} \;
		\prod_{e \in \mathrm{path}(u,v)} 
	\widehat{\bInf_e}
	 \cdot 
	\prod_{e \in \mathrm{path}(r,u)} 
	\widehat{\bInf_e}
	\cdot 
	\prod_{e^\prime \in \mathrm{path}(r,v)} 
	\widehat{\bInf_{e^\prime}} \enspace,
\end{align*}
where the first equality follows from the linearity of expectation, while the second equality follows from Lemma \ref{lem:ExpSigmaSigma}. Recalling that $u\wedge v$ denotes the common ancestor of $u$ and $v$ farthest from the root $r$, we can rewrite the above as 
\begin{align}\label{eq:FinalFormOfSecondMom}
	\mExp{F_h^2}_{\bmu} 
	=
	\sum_{u,v \in \Lambda} \;
	\prod_{e \in \mathrm{path}(u,v)} 
	\bInf_e^2
	\cdot 
	\prod_{{e^\prime} \in \mathrm{path}(r,u\wedge v)} 
	\bInf_{e^\prime}^2 \enspace.
\end{align}

We are now ready to prove \eqref{eq:EqOfRarioOfFhDen}. Per \eqref{eq:FinalFormOfSecondMom} we have that
\begin{align*}
		{\Exp\left[\mExp{F_h^2}_{\bmu}\right]}
	&=
		{\displaystyle\Exp\left[\sum_{u,v \in \Lambda}\;\left(\prod_{e \in \mathrm{path}(u,v)} \bInf_e^2\right)\left(\prod_{e^{\prime}\in\mathrm{path}(r,u\wedge v) }\bInf_{e^{\prime}}^2\right)\right]}
	 \\ 
	 &
 =
		{\displaystyle\sum_{u,v \in \Lambda}\Exp\left[\;\left(\prod_{e \in \mathrm{path}(u,v)} \bInf_e^2\right)\left(\prod_{e^{\prime}\in\mathrm{path}(r,u\wedge v) }\bInf_{e^{\prime}}^2\right)\right]}
	\\ 
	&=
		{\displaystyle\sum_{u,v \in \Lambda}\;\left(\prod_{e \in \mathrm{path}(u,v)}\Exp\left[ \bInf_e^2\right]\right)\left(\prod_{e^{\prime}\in\mathrm{path}(r,u\wedge v) }\Exp\left[\bInf_{e^{\prime}}^2\right]\right)}
\enspace,
\end{align*}
where the first equality follows from the linearity of expectation, and the second from the fact that the couplings $\{\bJ_e\}_{e\in E}$, (and thus, also $\{\bInf_e\}_{e\in E}$), are independent.  This concludes the proof of Proposition~\ref{prop:RatioOfFh}.

\section{Proof of Lemma \ref{lem:lowerBoundDtree}}\label{sec:ProofLemmaBound}
For integers $\Delta, h >0$, let now $T=(V,E)$ be the $\Delta$-ary tree rooted at $r$, and  $\bmu_{\beta,\phi}(\sigma)$ be the Gibbs measure on $T$ defined as in \eqref{def:XanaKaiXanaSpinGlass}. Let also $\Lambda$ be the set of vertices at distance $h$ from the root $r$. By Proposition~\ref{prop:RatioOfFh} we have that 
\begin{equation*}
	\frac{\left(\Exp\left[\mExp{F_h}_{\bmu_\Lambda^+} - \mExp{F_h}_{\bmu_\Lambda^-}\right]\right)^2}
		{4\Exp\left[\mExp{F_h^2}_{\bmu}\right]}
	= 
		\frac{\displaystyle\left(\sum_{v \in \Lambda} \prod_{e \in \mathrm{path}(r,v)}\Exp\left[ \bInf^2_e\right]\right)^2}
		{\displaystyle\sum_{u,v \in \Lambda}\;\left(\prod_{e \in \mathrm{path}(u,v)}\Exp\left[ \bInf_e^2\right]\right)\left(\prod_{e^{\prime}\in\mathrm{path}(r,u\wedge v) }\Exp\left[\bInf_{e^{\prime}}^2\right]\right)} \enspace.
\end{equation*}
Recalling  that
$
\textstyle
	\Delta_{\rm KS}
		=
		\Delta_{\rm KS}(\beta,\phi) 
		= 
		 \left( \mathbb{E}\left [   \left (\frac{1-\exp(\beta \bJ)}{1+\exp(\beta \bJ)}  \right)^2 \right] \right)^{-1}
		=
		 \left( \mathbb{E}\left [   \bInf_e^2 \right] \right)^{-1}
$,
we have that
\begin{align}
\nonumber
	\frac{\left(\Exp\left[\mExp{F_h}_{\bmu_\Lambda^+} - \mExp{F_h}_{\bmu_\Lambda^-}\right]\right)^2}
		{4\Exp\left[\mExp{F_h^2}_{\bmu}\right]}
		&=
		\frac{\displaystyle\left(\sum_{v \in \Lambda} \prod_{e \in \mathrm{path}(r,v)}		
			\Delta_{\rm KS}^{-1}\right)^2}
		{\displaystyle\sum_{u,v \in \Lambda}\;\left(\prod_{e \in \mathrm{path}(u,v)}
		\Delta_{\rm KS}^{-1}\right)\left(\prod_{e^{\prime}\in\mathrm{path}(r,u\wedge v) }
		\Delta_{\rm KS}^{-1}\right)}\\ \nonumber
		&=
		\frac{\displaystyle\left(\sum_{v \in \Lambda} \Delta_{\rm KS}^{-|\mathrm{path}(r,v)|}\right)^2}
		{\displaystyle\sum_{u,v \in \Lambda}\;\left(\Delta_{\rm KS}^{-|\mathrm{path}(u,v)|}\right)
		\left(\Delta_{\rm KS}^{-|\mathrm{path}(r,u\wedge v)|}\right)}\\ \label{eq:simpleSumObs}
		&=
		\frac{\displaystyle \Delta^{2h} \cdot\Delta_{\rm KS}^{-2h}}
		{\displaystyle\sum_{u,v \in \Lambda}\;
		\Delta_{\rm KS}^{|\mathrm{path}(r,u\wedge v)|-2h}}
		=
		\frac{\displaystyle \Delta^{2h} }
		{\displaystyle\sum_{u,v \in \Lambda}\;
		\Delta_{\rm KS}^{|\mathrm{path}(r,u\wedge v)|}}
\enspace,
\end{align}
where to get the first equality of \eqref{eq:simpleSumObs} we observe that $|\mathrm{path}(u,v)| = 2|\mathrm{path}(r,u\wedge v)|$. Writing now $\Lambda(\ell)$ for the vertices of $T$ at distance $0\le\ell \le h$ from the root $r$, and reorganising the sum in the denominator of \eqref{eq:simpleSumObs} with respect to the common ancestor $z = u\wedge v$, we get that
\begin{align}
\nonumber
	\sum_{u,v \in \Lambda}\;
		\Delta_{\rm KS}^{|\mathrm{path}(r,u\wedge v)|}
	&=
	\sum_{\ell=0}^{h}\;
	\sum_{z \in \Lambda(\ell)}\;
	\sum_{\substack{\;u,v \in \Lambda(h) \\ u\wedge v = z}}\;
		\Delta_{\rm KS}^{\ell}
		\\ \nonumber	&
		= 
	\sum_{\ell=0}^{h}\;
		\Delta_{\rm KS}^{\ell}
	\sum_{z \in \Lambda(\ell)}\;
	\sum_{\substack{\;u,v \in \Lambda(h) \\ u\wedge v = z}}\; 1 \\ \nonumber
	&=
	\sum_{\ell=0}^{h}\;
		\Delta_{\rm KS}^{\ell}
	\left[
	\Delta^{\ell}
	\cdot
	\left(\Delta^{h-\ell} -1\right)
	\cdot
	\Delta^{h-\ell}
	\right]\\ \nonumber
	&\le
	\sum_{\ell=0}^{h}\;
		\Delta_{\rm KS}^{\ell}
	\cdot
	\Delta^{\ell}
	\cdot
	{\Delta^{2(h-\ell)}}
\ 	= \
	\Delta^{2h}
	\sum_{\ell=0}^{h}\;
		\left(\frac{\Delta_{\rm KS}}
		{\Delta}\right)^{\ell}
		  \enspace,
\end{align}
which, due to our assumption that $\Delta=(1+\delta)\Delta_{\rm KS}$, for some $\delta > 0$, further simplifies to

\begin{align*}
	\sum_{u,v \in \Lambda}\;\Delta_{\rm KS}^{|\mathrm{path}(r,u\wedge v)|} 
		\le
	\Delta^{2h}
		\left(1-\frac{\Delta_{\rm KS}}{\Delta}\right)^{-1}
	= \Delta^{2h}\left(\frac{\delta}{1+\delta}\right)^{-1} \enspace.
\end{align*}
Plugging now the above into \eqref{eq:simpleSumObs} we finally get
\begin{align*}
	\frac{\left(\Exp\left[\mExp{F_h}_{\bmu_\Lambda^+} - \mExp{F_h}_{\bmu_\Lambda^-}\right]\right)^2}
		{4\Exp\left[\mExp{F_h^2}_{\bmu}\right]}
	\ge 
	\frac{\Delta^{2h}}{\Delta^{2h}\left(\frac{\delta}
		{1+\delta}\right)^{-1}}  
	= 
	\frac{\delta}
		{1+\delta}
\enspace.
\end{align*}
All the above complete the proof of Lemma~\ref{lem:lowerBoundDtree}.

\section{Proof of Theorem~\ref{thrm:GeneralGWtree} - reconstruction for the Galton-Watson Tree}

Let us briefly recall our setup. For  any real number $d>0$,  for $\beta>0$, for any distribution $\phi$ on $\mathbb{R}$, and any offspring distribution 
 $\zeta:\mathbb{Z}_{\geq 0}\to [0,1]$ with expectation $d$ and bounded second moment,   
 let $\Delta_{\rm KS}=\Delta_{\rm KS}(\beta,\phi)$ be defined as in \eqref{eq:DeltaKSBPhi}.
Let $\bT$ be the Galton-Watson tree with offspring distribution $\zeta$,  while let the Gibbs distribution 
$\bmu =\bmu_{\beta,\phi}$, defined as in \eqref{def:SpinGlass} ,  on the tree $\bT$. For an integer $h>0$, write $\pmb{\Lambda}$ for the set of vertices at distance $h$ from the root of $\bT$ (notice that $\pmb{\Lambda}$ is a random variable here). We also write $\bmu_{\pmb{\Lambda}}^{+}$ and $\bmu_{\pmb{\Lambda}}^{-}$ for the marginal of measure $\bmu$ on the set $\pmb{\Lambda}$, conditioned on the root of $\bT$ being $+$ and $-$, respectively.

We want to show that if $d>\Delta_{\rm KS}$, the distribution $\bmu_{\beta,\phi}$  exhibits reconstruction, i.e., 
\begin{align}\nonumber
\limsup_{h\to\infty}  \mathbb{E}_{\bT } \left[  \ 
\mathbb{E}_{\bmu } \left [  \lnorm \bmu_{\pmb{\Lambda}}^{+}(\cdot) -\bmu_{\pmb{\Lambda}}^{-}(\cdot)  \rnorm_{\rm TV} \ | \ \bT \right] 
\ \right]>0 \enspace.
\end{align}

We start by noticing that Theorem \ref{thm:TVtoSecond} can be extended to Galton-Watson trees. That is, we have that for any real function $G: \{\pm 1\}^{\pmb{\Lambda}}\mapsto \mathbb{R}$ defined on spin configurations of $\pmb{\Lambda}$, we have that
\begin{equation}\label{eq:GWTVtoSecond}
	\Exp_{\bT,\bmu}\left[\lnorm \bmu_{\pmb{\Lambda}}^{+}(\cdot) -\bmu_{\pmb{\Lambda}}^{-}(\cdot)  \rnorm_{\rm TV} \right]
	\ge
	\frac{\left(\Exp_{\bT,\bmu}\left[\mExp{G}_{\bmu_{\pmb{\Lambda}}^+} - \mExp{G}_{\bmu_{\pmb{\Lambda}}^-}\right]\right)^2}
		{4\Exp_{\bT,\bmu}\left[\mExp{G^2}_{\bmu}\right]} \enspace,
\end{equation}
In fact, the proof \eqref{eq:GWTVtoSecond} is almost identical to that of Theorem \ref{thm:TVtoSecond}, the only difference being that at the very last step of the proof, we apply Cauchy-Schwartz to an expression with an additional sum (due to $\Exp_{\bT}$). Hence, all it remains to do is to lower bound the rhs of \eqref{eq:GWTVtoSecond} away from zero.

From Proposition \ref{prop:RatioOfFh} and conditioning over the random tree $\bT$, we get that
\begin{align*}
	\frac{\left(\Exp_{\bT,\bmu}\left[\mExp{F_h}_{\bmu_{\pmb{\Lambda}}^+} - \mExp{F_h}_{\bmu_{\pmb{\Lambda}}^-}\right]\right)^2}
		{4\Exp_{\bT,\bmu}\left[\mExp{F_h^2}_{\bmu}\right]}
	&=
	\frac{\displaystyle\left(\Exp_{\bT}\left[\Exp_{\bmu}\left[\sum_{v \in {\pmb{\Lambda}}} \; \prod_{e \in \mathrm{path}(r,v)}\bInf^2_e \;\middle|\ \bT\right]\right]\right)^2}
		{\displaystyle\Exp_{\bT}\left[\Exp_{\bmu}\left[\sum_{u,v \in {\pmb{\Lambda}}}\left(\prod_{e \in \mathrm{path}(u,v)} \bInf_e^2\right)\left(\prod_{e^{\prime}\in\mathrm{path}(r,u\wedge v) }\bInf_{e^{\prime}}^2\right)\middle|\bT\right]\right]}  \enspace.
\end{align*}
Given a random tree $\bT$, the influences, $\bInf_e$, are independent. Moreover, recalling that   	
$\Delta_{\rm KS} = \left( \mathbb{E}_\mu[\bInf_e^2 ] \right)^{-1}$, (notice that $\Delta_{\rm KS}$ does not depend on $\bT$), we can further simplify the above as follows:
\begin{align*}
	\frac{\left(\Exp_{\bT,\bmu}\left[\mExp{F_h}_{\bmu_{\pmb{\Lambda}}^+} - \mExp{F_h}_{\bmu_{\pmb{\Lambda}}^-}\right]\right)^2}
		{4\Exp_{\bT,\bmu}\left[\mExp{F_h^2}_{\bmu}\right]}
	&= 
	\frac{\displaystyle\left(\Exp_{\bT}\left[\sum_{v \in {\pmb{\Lambda}}(h)} \; \Delta_{\rm KS}^{-h}\right]\right)^2}
		{\displaystyle\Exp_{\bT}\left[\sum_{u,v \in {\pmb{\Lambda}}(h)}\;\Delta_{\rm KS}^{|\mathrm{path}(r,u\wedge v)|-2h}\right]} 
	=
	\frac{\displaystyle\left(\Exp_{\bT}\left[\;|{\pmb{\Lambda}}(h)|\; \right]\right)^2}
		{\displaystyle\Exp_{\bT}\left[\sum_{u,v \in {\pmb{\Lambda}}(h)}\;\Delta_{\rm KS}^{|\mathrm{path}(r,u\wedge v)|}\right]} \enspace.
\end{align*}

Since from this point on we are left only with expectations with respect to $\bT$, we drop the subscript in~$\Exp$.
Reorganising the sum of in the denominator in the last equation above, similarly to 
the proof of Lemma~\ref{lem:lowerBoundDtree}, and writing ${\pmb{\Lambda}}_z(\ell)$ for the descendants of $z$ at distance $\ell$ (${\pmb{\Lambda}}$ without subscript refers to descendants of the root $r$), we get that
\begin{align}\label{eq:SimilarReorg}
\Exp\left[\sum_{u,v \in {\pmb{\Lambda}}(h)}\Delta_{\rm KS}^{|\mathrm{path}(r,u\wedge v)|}\right]
	&=
	\sum_{\ell=0}^{h}\Delta_{\rm KS}^{\ell} \cdot
	{\Exp\left[
	\sum_{z \in {\pmb{\Lambda}}(\ell)}
	\sum_{w,q \in {\pmb{\Lambda}}_z(1) } 
	|{\pmb{\Lambda}}_{w}(h-\ell-1) |\cdot |{\pmb{\Lambda}}_{q}(h-\ell-1) |
	\right]} \enspace.
\end{align} 
Let us now focus on the expectation in the r.h.s. of the above equation. In particular, we invoke the law of total expectation conditioning on the set ${\pmb{\Lambda}}(\ell)$, comprised of the vertices at distance $\ell$ from the root
\begin{align}
\nonumber
	&{\Exp\left[
		\sum_{z \in {\pmb{\Lambda}}(\ell)}
		\sum_{w,q \in {\pmb{\Lambda}}_z(1) } 
		|{\pmb{\Lambda}}_{w}(h-\ell-1) |\cdot |{\pmb{\Lambda}}_{q}(h-\ell-1) |\right]} \\ \nonumber
		&=\Exp \left[{\Exp\left[
		\sum_{z \in {\pmb{\Lambda}}(\ell)} \sum_{w,q \in {\pmb{\Lambda}}_z(1) } 
		|{\pmb{\Lambda}}_{w}(h-\ell-1) |\cdot |{\pmb{\Lambda}}_{q}(h-\ell-1) |\middle|{\pmb{\Lambda}}(\ell)\right]}\right]\\ \label{eq:FirstConditioning}
	&=\Exp \left[|{\pmb{\Lambda}}(\ell)|\cdot{\Exp\left[
		\sum_{w,q \in {\pmb{\Lambda}}(1) } 
		|{\pmb{\Lambda}}_{w}(h-\ell-1) |\cdot |{\pmb{\Lambda}}_{q}(h-\ell-1) |\right]}\right]\enspace,
\end{align}
where \eqref{eq:FirstConditioning} follows from the linearity of expectation, and the fact that the offsprings of each vertex are identically distributed (and hence, the random variable ${\pmb{\Lambda}}_z(1)$ coincides with ${\pmb{\Lambda}}(1)$, for all vertices $z$). We now estimate the inner expectation of \eqref{eq:FirstConditioning} conditioning on ${\pmb{\Lambda}}(1)$.
\begin{align}
\nonumber
  &\Exp\left[
		\sum_{w,q \in {\pmb{\Lambda}}(1) } 
		|{\pmb{\Lambda}}_{w}(h-\ell-1) |\cdot |{\pmb{\Lambda}}_{q}(h-\ell-1) |\right]\\ \nonumber
  &=\Exp\left[\Exp\left[
		\sum_{w,q \in {\pmb{\Lambda}}(1) } 
		|{\pmb{\Lambda}}_{w}(h-\ell-1) |\cdot |{\pmb{\Lambda}}_{q}(h-\ell-1) |\middle|{\pmb{\Lambda}}(1)\right]\right]\\
\label{eq:SecondConditioning}
  &=\Exp\left[|{\pmb{\Lambda}}(1)|\cdot(|{\pmb{\Lambda}}(1)|-1) \cdot\left(\Exp\Big[
		|{\pmb{\Lambda}}(h-\ell-1) |\Big]\right)^2\right] \\ \label{eq:TeleiomenoExpDenominator}
  &\le\Exp\left[|{\pmb{\Lambda}}(1)|^2\right] \cdot\left(\Exp\Big[
		|{\pmb{\Lambda}}(h-\ell-1) |\Big]\right)^2		
		\enspace,
\end{align}
where \eqref{eq:SecondConditioning} follows from the linearity of expectation, and the fact that the offsprings of each vertex are independent (and thus, the inner expectation of products becomes the product of the corresponding expectations), and identically distributed (and hence, ${\pmb{\Lambda}}_w(h-\ell-1) = {\pmb{\Lambda}}_q(h-\ell-1)={\pmb{\Lambda}}(h-\ell-1)$). 

Putting them all together, we have that \eqref{eq:SimilarReorg}, \eqref{eq:FirstConditioning},\eqref{eq:SecondConditioning}, and \eqref{eq:TeleiomenoExpDenominator} yield

\begin{align*}
	&\frac{\left(\Exp_{\bT,\bmu}\left[\mExp{F_h}_{\bmu_{\pmb{\Lambda}}^+} 
			- \mExp{F_h}_{\bmu_{\pmb{\Lambda}}^-}\right]\right)^2}
		  {4\Exp_{\bT,\bmu}\left[\mExp{F_h^2}_{\bmu}\right]}
				\ge
	\frac{\displaystyle\left(\Exp\left[\;|{\pmb{\Lambda}}(h)|\; \right]\right)^2}
		{\displaystyle\sum_{\ell=0}^{h}\Delta_{\rm KS}^{\ell}
		\cdot\Exp\Big[|{\pmb{\Lambda}}(\ell)|\Big]
		\cdot
		\Exp
		\Big[
		|{\pmb{\Lambda}}(1)|^2 
		\Big]\cdot
		\left(\Exp
		\Big[|{\pmb{\Lambda}}(h-\ell-1) |\Big]\right)^2} \enspace.
\end{align*}

Due to the fact that the offsprings of vertices in $\bT$ are i.i.d., we observe that for any $\ell\ge0$, we have that $\Exp[|{\pmb{\Lambda}}(\ell) | = (\Exp[\zeta])^\ell = d^\ell$, and thus, we can further simplify the above as follows
		
\begin{align*}
	\frac{\left(\Exp_{\bT,\bmu}\left[\mExp{F_h}_{\bmu_{\pmb{\Lambda}}^+} 
			- \mExp{F_h}_{\bmu_{\pmb{\Lambda}}^-}\right]\right)^2}
		  {4\Exp_{\bT,\bmu}\left[\mExp{F_h^2}_{\bmu}\right]}
				&\ge
	\frac{\displaystyle\left(\Exp\left[\;|{\pmb{\Lambda}}(h)|\; \right]\right)^2}
		{\displaystyle\sum_{\ell=0}^{h}\Delta_{\rm KS}^{\ell}
		\cdot\Exp\Big[|{\pmb{\Lambda}}(\ell)|\Big]
		\cdot
		\Exp
		\Big[
		|{\pmb{\Lambda}}(1)|^2 
		\Big]\cdot
		\left(\Exp
		\Big[|{\pmb{\Lambda}}(h-\ell-1) |\Big]\right)^2}\\
			&=
	\frac{\left(\Exp\left[\zeta\right]\right)^{2h}}
	{\displaystyle\sum_{\ell=0}^{h}
	\Delta_{\rm KS}^{\ell}\cdot
	\left(\Exp
	\left[\zeta\right]\right)^\ell\cdot
	\Exp
	\left[\zeta^2\right]\cdot
		\left(\Exp
	\left[\zeta\right]\right)^{2h-2\ell-2}}\\
	&=\frac{1}
	{\displaystyle 
	\frac{\Exp\left[\zeta^2\right]}
	{d^{2}}\cdot
	\sum_{\ell=0}^{h}\left(
	\frac{\Delta_{\rm KS}}
	{d}\right)^\ell
	}\enspace.
\end{align*}

Per our hypothesis, $\Exp \left[\zeta^2\right]<\infty$, and thus, $\Exp \left[\zeta^2\right] \le Md^{2}$, for some bounded number  $M>0$. Moreover, we have that $\Delta_{\rm KS} < d$, and thus, there exist a $\delta>0$, such that $\Delta_{\rm KS} (1+\delta)=d$. With that in mind, we further bound the above as

\begin{align*}
	\frac{\left(\Exp_{\bT,\bmu}\left[\mExp{F_h}_{\bmu_{\pmb{\Lambda}}^+} 
			- \mExp{F_h}_{\bmu_{\pmb{\Lambda}}^-}\right]\right)^2}
		  {4\Exp_{\bT,\bmu}\left[\mExp{F_h^2}_{\bmu}\right]}
	&\ge\frac{1}
	{\displaystyle 
	\frac{\Exp\left[\zeta^2\right]}
	{d^{2}}\cdot
	\sum_{\ell=0}^{h}\left(
	\frac{\Delta_{\rm KS}}
	{d}\right)^\ell}
	\ge
	\frac{\delta}
		{M(1+\delta)} >0
	\enspace.
\end{align*}
This concludes the proof of the reconstruction claim of Theorem~\ref{thrm:GeneralGWtree}.

\bibliographystyle{plainurl}
\bibliography{PapersReconstruction}

\appendix

\section{Equivalence of Indicator and Product Gibbs distribution}\label{sec:IndProdEquiv}

Let $G=(V,E)$ be a graph, and let $\{J_e : e \in E\}$ be arbitrary couplings over the edges of $G$. For $\beta > 0$, let us write $\mu_I$, and $\mu_P$ for the Gibbs distributions over $\{\pm 1\}^V$, defined by the indicator, and product formulation, respectively. That is, for every $\sigma \in \{\pm 1\}^V$, we have 
\begin{align*}
\mu_I(\beta;\sigma) &\propto 
\exp\left( \beta \cdot\sum_{\{u,w\}\in E} {\bf 1}\{\sigma(u)=\sigma(w)\} \cdot \bJ_{\{u,w\}}   \right)\enspace, \\
\mu_P(\beta;\sigma) &\propto 
\exp\left( \beta \cdot\sum_{\{u,w\}\in E} \sigma(u)\sigma(w) \cdot \bJ_{\{u,w\}}   \right)\enspace.
\end{align*}
We will prove that $\mu_P(\beta;\sigma) = \mu_I(2\beta;\sigma)$, for every $\sigma \in \{\pm 1\}^V$. Indeed, let $\sigma \in \{\pm 1\}^V$ be arbitrary, then 
\begin{align*}
\mu_P(\beta;\sigma) &\propto 
\exp\left( \beta \cdot\sum_{\{u,w\}\in E} \sigma(u)\sigma(w) \cdot \bJ_{\{u,w\}}   \right)\\
&=
\exp\left( \beta \cdot\sum_{\{u,w\}\in E} (2\cdot {\bf 1}\{\sigma(u)=\sigma(w)\} -1)\cdot \bJ_{\{u,w\}}   \right)\\
&=
\exp\left( 2\beta \cdot\sum_{\{u,w\}\in E}  {\bf 1}\{\sigma(u)=\sigma(w)\} \cdot \bJ_{\{u,w\}}   \right)\cdot \exp\left(- \beta \cdot\sum_{\{u,w\}\in E}\bJ_{\{u,w\}}   \right)\\
&\propto
\exp\left( 2\beta \cdot\sum_{\{u,w\}\in E}  {\bf 1}\{\sigma(u)=\sigma(w)\} \cdot \bJ_{\{u,w\}}   \right) 
=
\mu_I(2\beta;\sigma)
\enspace.
\end{align*}
Since $\mu_P$, $\mu_I$, are probability measures, we conclude that $\mu_P(\beta;\sigma) = \mu_I(2\beta;\sigma)$, as desired.

\section{KS-Bound Derivation}\label{sec:KSMatrix}

First, note that since $\bM$ is symmetric, $\bM \otimes \bM $ must be symmetric as well. In particular, we have that 

\begin{equation}\label{eq:TensorBMat}
\bM\otimes\bM = 
\frac{1}{\left(1+e^{\beta \bJ}\right)^2}
		\begin{pmatrix} 
			e^{2\beta \bJ}	& 	e^{\beta\bJ}	& 	e^{\beta \bJ}	&	1		      \\[2pt]
			e^{\beta \bJ}	& 	e^{2\beta\bJ}	& 	1	    		&	e^{\beta\bJ}  \\[2pt]
			e^{\beta \bJ}	& 	1			& 	e^{2\beta \bJ}	&	e^{\beta\bJ}  \\[2pt]
			1	 		& 	e^{\beta\bJ}	& 	e^{\beta \bJ}	&	e^{2\beta \bJ}
			\end{pmatrix}
				\enspace,
\end{equation}
It is also easy to check that the for any matrix with the same pattern on its entries we have the following

\begin{observation}\label{ob:LinearSpec}
The spectrum of every $4\times 4$ matrix, $B$, of the following form
\begin{equation}\label{eq:easyForm}
B = 
				\begin{pmatrix} 
					a	& 	b	& 	b	&	c\\
					b	& 	a	& 	c	&	b\\
					b	& 	c	& 	a	&	b\\
					c	& 	b	& 	b	&	a\\
				\end{pmatrix}
				\enspace, \text{ with } a, b, c, d \in \mathbb{R}\enspace, 
\end{equation}
is precisely $\{\{\lambda_1 := (a+2b+c),\; \lambda_2 :=(a-c),\; \lambda_3 :=(a-c)\; \lambda_4 :=(a-2b +c),\}\}$. In particular, every eigenvalue of $B$ is a linear combination of its elements.
\end{observation}

Note that both $\bM\otimes\bM$, and $\Exp[\bM\otimes\bM]$, are of the form \eqref{eq:easyForm}. In the following lemma we show that
Observation \ref{ob:LinearSpec} allows us to change the order of averaging and taking eigenvalues of $\bM\otimes\bM$.

\begin{lemma}\label{lem:CommuteEL}
Let $\lambda_1,\lambda_2, \lambda_3, \lambda_4$, be as in Observation \ref{ob:LinearSpec}. Then, for every $1 \le k\le 4$ we have that
\begin{equation}
\lambda_k\left(\Exp\left[\bM\otimes\bM\right]\right) = \Exp\left[\lambda_k\left(\bM\otimes\bM\right)\right] \enspace.
\end{equation}
\end{lemma}

\begin{proof}
Since both $\bM\otimes\bM$, and $\Exp[\bM\otimes\bM]$, are of the form \eqref{eq:easyForm}, each $\lambda_k$ is a linear combination of their entries, and thus, the result follows by the linearity of expectation.
\end{proof}

Let us now recall that equation \eqref{eq:DeltaKSBPhi} defins $\Delta_{\rm KS}$ as follows
\begin{equation*}
\Delta_{\rm KS} =
\left( \max_{x\in \cE: \lnorm x \rnorm=1}\langle \Xi x, x\rangle\right)^{-1} \enspace,
\end{equation*}
where 
$
\cE = \left\{ z\in \mathbb{R}^{\cA}\otimes \mathbb{R}^{\cA}: \forall y\in \mathbb{R}^{\cA} 
 \langle z, {\bf 1}\otimes y\rangle=\langle z, y\otimes{\bf 1}\rangle=0\right\}
$, and $\Xi = \Exp[\bM\otimes\bM]$. Since $\Xi$ is of the form \eqref{eq:easyForm}, and in particular symmetric, it is easy to argue, e.g. using Courant-Fisher theorem, that
the solution to the maximisation in \eqref{eq:DeltaKSBPhi} must be 
\begin{equation}
\max_{x\in \cE: \lnorm x \rnorm=1}\langle \Xi x, x\rangle = \lambda_4(\Xi) \enspace.
\end{equation}
Using now Lemma~\ref{lem:CommuteEL} we get that $\Delta_{\rm KS}= (\lambda_4(\Xi))^{-1} = (\Exp\left[\lambda_4\left(\bM\otimes\bM\right)\right])^{-1}$. Substituting the entries of $\bM\otimes\bM$ from \eqref{eq:TensorBMat}, yields 
$
\Delta_{\rm KS}= 
\textstyle \left( \mathbb{E}\left [   \left (\frac{1-\exp(\beta \bJ)}{1+\exp(\beta \bJ)}  \right)^2 \right] \right)^{-1}
$, as desired.

\section{Proof of Lemma \ref{lem:UpDwnUpperBound}}\label{sec:lem:UpDwnUpperBound}

\begin{proof}
First, let us recall that we denote with $\Lambda$ the set of vertices at distance $h$ from the root $r$. Also, for $s\in \{\pm1\}$ and $\tau \in \{\pm 1\}^\Lambda$, we write $\mu_r^{\Lambda,\tau}$, and $\mu_\Lambda^s$, for the marginal of $\mu$ on the root, conditioned on $\bsigma(\Lambda) = \tau$, and the marginal of $\mu$ on the the set $\Lambda$, conditioned on $\bsigma(r) = s$, respectively.  We now have that 
\begin{align} \nonumber
	\left\|\mu_\Lambda^+(\cdot) - \mu_\Lambda^-(\cdot)\right\|_{\rm TV} 
		&= 
	\frac{1}{2} \sum_{\tau \in \{\pm1\}^{\Lambda}} 
	\left|\mu_\Lambda^+(\tau) - \mu_\Lambda^-(\tau) \right|\\ \label{eq:doingBayes}
	&= \frac{1}{2} \sum_{\tau\in \{\pm1\}^{\Lambda}} 
	\left|\mu_r^{\Lambda,\tau}\left(+1\right) 
		\frac{\mu(\bsigma(\Lambda)=\tau)}{\mu(\bsigma(r) = +1)}
		- \mu_r^{\Lambda,\tau}\left(- 1\right)
		\frac{\mu(\bsigma(\Lambda)=\tau)}{\mu(\bsigma(r) = -1)}\right|\\ \label{eq:doing+-half}
	&= \sum_{\tau \in \{\pm1\}^{\Lambda}} 
	\left|\mu_r^{\Lambda,\tau}\left( +1\right)
		-\mu^{\Lambda,\tau }_r\left(-1\right)\right| \cdot 
		{\mu(\bsigma(\Lambda)=\tau)}\\ \nonumber
	&= \Exp_{\tau\sim \mu_\Lambda}\left[
		\left|\mu_r^{\Lambda,\tau}\left( +1\right)
		-\mu^{\Lambda,\tau }_r\left(-1\right)\right| 
		\right]\enspace,
\end{align}
where \eqref{eq:doingBayes} follows from Bayes' rule, and \eqref{eq:doing+-half} is due to the 
fact that
$
\mu(\bsigma(r) =-1)= \mu(\bsigma(r) =+1)={1}/{2}
$.
Next, we observe that 
\begin{align}
\Exp_{\tau  \sim \mu_\Lambda}\left[\left(\mu_r^{\Lambda,\tau}(+1)\right)^2\right]  \nonumber
&=
\Exp_{\tau  \sim \mu_\Lambda}\left[\left(1-\mu_r^{\Lambda,\tau}(-1)\right)^2\right]\\ \nonumber
&= 1 -2\Exp_{\tau  \sim \mu_\Lambda}\left[\mu_r^{\Lambda,\tau}(-1)\right]
+\Exp_{\tau \sim \mu_\Lambda}\left[\left(\mu_r^{\Lambda,\tau}(-1)\right)^2\right]\\\nonumber
&= 1-2\cdot\mu(\bsigma(r)=+1)
+\Exp_{\tau \sim \mu_\Lambda}\left[\left(\mu_r^{\Lambda,\tau}(-1)\right)^2\right]\\ \label{eq:apodw}
&=\Exp_{\tau \sim \mu_\Lambda}\left[\left(\mu_r^{\Lambda,\tau}(-1)\right)^2\right] \enspace.
\end{align}
Using the above, we also get that
\begin{align}
\nonumber
\left\|\mdownup^{+}(\cdot)-\mdownup^{-}(\cdot)\right\|_{\rm TV}  &=
	\left|\mdownup^{+}(+1) 
		- \mdownup^{-}(+1)\right|\\ \label{eq:defofupdaownc}
	&= 	
	\left|\sum_{\tau \in \{\pm 1\}^\Lambda}
	\mu^+_{\Lambda}(\tau) \mu_r^{\Lambda,\tau}(+1)
		- \mu^-_{\Lambda}(\tau) \mu_r^{\Lambda,\tau}(+1)\right|\\ \label{eq:justBayes}
	&= 	
	\left|\sum_{\tau\in \{\pm1\}^\Lambda}
	\left(
	\mu_r^{\tau,\Lambda}(+1)
		-\mu_r^{\tau,\Lambda}(-1)
	\right)\cdot
		2\cdot\mu\left(\bsigma(\Lambda)=\tau\right)\cdot \mu_r^{\Lambda,\tau}(+1)\right|\\ \nonumber
	&=
		\left|2\cdot\Exp_{\tau \sim \mu_\Lambda}\left[\left(\mu_r^{\Lambda,\tau}(+1)\right)^2\right]
		-2\cdot\Exp_{\tau \sim \mu_\Lambda}\left[\mu_r^{\Lambda,\tau}(+1)\mu_r^{\Lambda,\tau}(-1)\right]
	\right|\\ \label{eq:zori}
	&=
	\left|\Exp_{\tau \sim \mu_\Lambda}\left[\left(\mu_r^{\Lambda,\tau}(+1) - \mu_r^{\Lambda,\tau}(+1)\right)^2\right]
	\right|\\ \nonumber
	&= \Exp_{\tau \sim \mu_\Lambda}\left[	\left|
	\mu_r^{\Lambda,\tau}(+1)
		-\mu_r^{\Lambda,\tau}(-1)\right|^2\right] \enspace,
\end{align}
where \eqref{eq:defofupdaownc} follows from \eqref{eq:downupis},  while \eqref{eq:justBayes} is due to the Bayes' rule, and the fact that
$
\mu(\bsigma(r) =-1)= \mu(\bsigma(r) =+1)={1}/{2}
$. Finally,
we get \eqref{eq:zori} from the observation \eqref{eq:apodw}.
The result now follows from the Cauchy-Schwartz inequality.
\end{proof}

\end{document}